\documentclass[cmafonts,final]{statsoc}

\usepackage{amsmath}
\usepackage{amssymb}
\usepackage[font=small,labelfont=bf]{caption}
\usepackage{floatrow}
\usepackage{natbib}
\usepackage{soul}
\usepackage{xcolor}
\usepackage{xfrac}
\usepackage[toc,page]{appendix}
\usepackage{blowup}
\blowUp{paper=a4}

\newcommand\mcD{{\mathcal D}}

\newcommand\mcH{{\mathcal H}}

\newcommand\mcL{{\mathcal L}}

\newcommand\bm{{\bf m}}

\newcommand\mcN{{\mathcal N}}

\newcommand\mcS{{\mathcal S}}

\newcommand\bx{{\bf x}}

\newcommand\cH{{\mathcal H}}

\newcommand\mbE{{\mathbb E}}

\newcommand\mbR{{\mathbb R}}
\newcommand\mbP{{\mathbb P}}

\DeclareMathOperator\Cov{Cov}
\DeclareMathOperator\Span{span}

\newcommand{\iid}{\stackrel{iid}{\sim}}

\newtheorem{theorem}{Theorem}

\newtheorem{lemma}{Lemma}

\newtheorem{corol}{Corollary}

\newtheorem{assumption}{Assumption}

\title[Confidence Regions and Bands for FDA]{A Geometric Approach to Confidence Regions \\ and Bands for Functional Parameters}

\author{Hyunphil Choi}
\address{The Pennsylvania State University, University Park, PA, USA.}
\author[Choi and Reimherr]{Matthew Reimherr}\coaddress{Matthew Reimherr, Department of Statistics, The Pennsylvania State University, 411 Thomas Building, University Park, PA 16802, USA. \textit{mreimherr@psu.edu}}
\address{The Pennsylvania State University, University Park, PA, USA.}

\begin{document}

\begin{abstract}
	Functional data analysis, FDA, is now a well established discipline of statistics, with its core concepts and perspectives in place. Despite this, there are still fundamental statistical questions which have received relatively little attention. One of these is the systematic construction of confidence regions for functional parameters.  This work is concerned with developing, understanding, and visualizing such regions. We provide a general strategy for constructing confidence regions in a real separable Hilbert space using hyper-ellipsoids and hyper-rectangles. We then propose specific implementations which work especially well in practice. They provide powerful hypothesis tests and useful visualization tools without using any simulation.
	We also demonstrate the negative result that nearly all regions, including our own, have \textit{zero-coverage} when working with empirical covariances.  To overcome this challenge we propose a new paradigm for evaluating confidence regions by showing that the distance between an estimated region and the desired region (with proper coverage) tends to zero faster than the regions shrink to a point. We call this phenomena \textit{ghosting} and refer to the empirical regions as \textit{ghost} regions. 
	We illustrate the proposed methods in a simulation study and an application to fractional anisotropy tract profile data. 
\end{abstract}


\section{Introduction}

Functional data analysis, FDA, is a branch of statistics whose foundational work goes back at least two decades.  Its development and application has seen a precipitous increase in recent years due to the emergence of new data gathering technologies which incorporate high frequency sampling.  Fundamentally, FDA is concerned with data which can be viewed as samples of curves, images, shapes, or surfaces.  
While FDA is now a well established discipline with its core tools and concepts in place, there are still fundamental questions that have received relatively little attention.  This work is concerned with developing, understanding, and visualizing confidence regions for functional data, a fundamental statistical concept which has received little attention in the FDA literature.  Our approach is geometric in that we start with general hyper-ellipsoids and hyper-rectangles and show how they can be tailored to become proper confidence regions.  A distinguishing feature of functional confidence regions is that, when the covariance of the estimator is estimated, nearly all confidence regions turn out to have \textit{zero-coverage} for the parameter; 
this is primarily due to the infinite dimensional nature of the parameter.  However, we demonstrate how most of these regions 
are very close to the proper regions with respect to Hausdorff distance. Such an issue does not occur in multivariate statistics and is a distinct feature of FDA.
We refer to this phenomenon as \textit{ghosting}, namely, that while one uses a confidence regions with  zero-coverage, they can be shown to be arbitrarily close to a proper confidence region with the desired coverage. Of course, for these \textit{ghost} regions to be useful, these distances must decrease faster than the rate at which the regions shrink down to a point.

Forming a confidence region for a functional parameter can equivalently be thought of as forming a confidence region for an infinite dimensional parameter. To see why this is a challenge, consider a classic multivariate confidence region.  Suppose that $\theta \in \mbR^p$ and we have an estimator, $\hat \theta$, which is multivariate normal, $\hat \theta \sim \mcN_p(\theta, \Sigma)$.  The classic approach to forming a $1-\alpha$ confidence region, $G_\alpha$, is to take the following ellipse
\begin{align}
G_\alpha = \{ x \in \mbR^p : (\hat \theta - x)^\top  \Sigma^{-1} (\hat \theta - x)  \leq \xi_\alpha \}. 
\label{e:mult_conf}
\end{align}
The constant $\xi_\alpha$ is chosen so that the region achieves the proper coverage; when $\Sigma$ is known it is taken as the quantile of a $\chi^2$, while when $\Sigma$ is estimated it is taken from an $F$. For $p$ very large, at least two things happen: (1) the inversion of $\Sigma$ becomes very unstable due to small eigenvalues and (2) the constant $\xi_\alpha$ becomes very large.  In fact, a naive functional analog would require $\xi_\alpha = \infty$ and the sample covariance operator would not even be invertible.

To address this problem in the functional case, there have been at least two  main approaches. The first approach is to develop confidence bands via simulation techniques \citep{degras:2011,cao:2012,zheng:2014,cao:2014}. 
These methods work quite well, but they shift the focus from Hilbert spaces, usually $L^2[0,1]$, to Banach spaces, such as $C[0,1]$.  Given that Hilbert spaces are the foundation of the large majority of theory and methods for FDA, it is important to have a procedure which is based on Hilbert spaces.  A more minor issue is that such bands, after taking into account point-wise variability, are usually built upon using a constant threshold across all time points.  For most practical purposes, this works well, but for objects with highly complex intra-curve dependencies, it could be useful to adjust the bands. For example, in areas with high positive within curve correlation, the bands can be made narrower, and in areas with very weak correlation they should be made wider.  Finally, simulation based approaches are computationally intensive, especially if one wants to invert the procedure to find very small p-values, which is very common in genetic studies, or increase evaluation points on the domain, i.e. work on a finer grid.   

The second approach is based on functional principal component analysis, FPCA \citep{ramsay:silverman:2005,yao:muller:wang:2005JASA,goldsmith:2013}. There one uses FPCA for dimension reduction and builds multivariate confidence ellipses, which can be turned into bands using Scheff\'e's method.  As we will show, this procedure produces ellipses which have \textit{zero-coverage}.
The dimension reduction inherently clips part of the parameter, meaning that the true parameter will never lie in the region.  As a simple illustration, imagine trying to capture a two dimensional parameter with an ellipse versus a line segment. The probability of capturing the parameter with a random ellipse can usually be well controlled, but any random line segment will fail to capture the parameter with probability one.  Additionally, the bands formed from these ellipses, depend heavily on the number of FPCs used.

The paper and its contributions are organized as follows.  
In Section \ref{section:ConfRegion} we present a new geometric approach to constructing confidence regions in real separable Hilbert spaces using hyper-ellipsoids (Section \ref{s:ellipse}) and  hyper-rectangles (Section \ref{s:rectangle}).  We show how to transform confidence hyper-ellipses into confidence bands and propose a specific ellipse which gives the smallest average squared width when turned into a band (Section \ref{s:bands}).  Simulations in Section \ref{section:Simulation} suggest that this ellipse is an excellent starting point for practitioners.  We also propose a visualization technique using rectangular regions (Section \ref{s:rectangle-visual}).  
In Section \ref{section:EstConfRegion}, we detail issues involved in using estimated/empirical versions based on estimated covariances.  As a negative result, we will show that nearly all empirical regions have zero--coverage.  However, we justify using these regions in practice by introducing the concept of \textit{`ghosting'}: using regions with deficient coverage as estimates for regions with proper coverage. %
Lastly, in Sections \ref{section:Simulation} and \ref{s:dti}, we provide a 
simulation study and an application to \texttt{DTI} data in the \texttt{R} package \texttt{refund}  \citep{goldsmith:2012a,goldsmith:2012b}.

\section{Constructing Functional Confidence Regions} \label{section:ConfRegion}
Throughout this paper we consider a general functional parameter $\theta \in \mcH,$ where $\mcH$ is a real separable Hilbert space with inner product $\langle \cdot, \cdot \rangle$.  We assume that we have an estimator $\hat{\theta} \in \mcH$ which is asymptotically Gaussian in $\mcH$ in the sense that $\sqrt N (\hat \theta - \theta) \overset{d}{\to} \mcN(0,  C_\theta)$, where $N$ is the sample size and $C_\theta$ is a covariance operator that can be estimated.   
Although multivariate confidence regions \eqref{e:mult_conf} are ellipsoids, this geometric shape is a by--product of using quadratic forms.  Here, however, we take the opposite approach. We first define the desired geometric shape and then demonstrate how to adjust the region to achieve the desired confidence level. 
Recall that $1-\alpha$ (asymptotic) confidence region $G_{\hat{\theta}}$ for $\theta \in \mcH$ is a random subset of $\mcH$ which satisfies $\mbP(\theta \in G_{\hat{\theta}}) \to 1-\alpha$.  We make the following assumption to simplify arguments.
\begin{assumption} \label{a:normal}
	Assume that $\sqrt{N}(\hat \theta - \theta) \overset{d}{\to} \mcN(0, C_\theta)$, that is, is asymptotically Gaussian in $\mcH$ with mean zero and covariance operator $C_\theta$.
\end{assumption}
	Assumption \ref{a:normal} is fairly weak and satisfied by many methods for dense functional data including mean estimation \citep{degras:2011}, covariance estimation \citep{zhang:wang:2016}, eigenfunction/value estimation \citep{kokoszka:reimherr:2013b}, and function-on-scalar regression \citep{ReNi:2014}.  To achieve such a property one needs that (i) the bias of the estimate is asymptotically negligible and that (ii) the estimate is \textit{tight} so that convergence in distribution occurs in the strong topology. While these two conditions are often satisfied, there are still many FDA settings where they are not.  The bias can usually be shown to be asymptotically negligible when the number of points sampled per curve is greater than $N^{1/4}$ \citep{li:hsing:2010,cai:2011,zhang:wang:2016}.  Thus our approach will not work for sparse FDA settings. The tightness assumption is often violated when estimates stem from ill-posed inverse problems. For example, in scalar-on-function regression, typical slope estimates are not tight and not asymptotically normal in the strong topology \citep{cardot:2007}.  

The backbone of our construction, and many other FDA methods, is the Karhunen-Lo\`eve, KL, expansion which gives 
\begin{align}
\label{e:KL}
\sqrt{N} (\hat{\theta} - \theta) = \sum_{j=1}^\infty \sqrt{\lambda_j}Z_jv_j,
\end{align}
where $\{\lambda_j\}$ and $\{v_j\}$ are eigenvalues and eigenfunctions, respectively, of $C_\theta$, and $\{Z_j\}$ are uncorrelated with mean zero and unit variance.  We note that this expansion holds for any random element in $\mcH$ with a finite second moment and that the infinite sum converges in $\mcH$.
In the next subsections we discuss two types of regions which exploit this expansion.  The first is a hyper-ellipse which is, as in the multivariate case, much easier to construct.  The second is a hyper-rectangle which is not mentioned as often in the multivariate literature due to the complexity of its form. However, in Section \ref{section:Simulation} we will show that in some settings the hyper-rectangle can outperform the ellipse and is usually much more interpretable.

\subsection{Hyper-Ellipsoid Form} \label{s:ellipse}
A hyper-ellipse in any Hilbert space can be defined as follows.  One needs a center, $m \in \mcH$, axes, $e_1, e_2, \cdots $, which are an orthonormal basis for $\mcH$, and a radius for each axis, $r_1, r_2, \cdots$.  The ellipse is then given by
\[
\left\{ h \in \mathcal{H} : \sum_{j=1}^{\infty} \frac{\langle h - m, e_j \rangle^2}{r_j^2} \leq 1 \right\}.
\]
We note that this definition makes sense even when $r_j = 0$ or $\infty$.  In the former one is saying that the radius in that direction is zero or `closed', while in the latter one is saying that it is infinite or `opened'.
Since our aim is to construct a confidence region for $\theta$, we will replace the arbitrary axes above with the eigenfunctions $\{v_j\}$ and the center with $\hat \theta$, to get
\[
E_{\hat{\theta}} := \left\{ h \in \mathcal{H} : \sum_{j=1}^{\infty} \frac{\langle h - \hat{\theta}, v_j \rangle^2}{r_j^2} \leq 1 \right\}.
\]
This hyper-ellipsoid will be a $1-\alpha$ confidence region for $\theta$  if we find $\{r_j\}$ which give
\[
\mbP(\theta \in E_{\hat{\theta}}) 
= \mbP\left( \sum_{j=1}^{\infty} \frac{\langle \theta - \hat{\theta}, v_j \rangle^2}{r_j^2} \leq 1 \right)
\to 1-\alpha.
\]
Note that there are actually infinitely many options for $\{r_j\}$ but not all of them lead to `nice' regions. 
We decompose $r_j^2 = N^{-1} \xi c_j^2$, where $\{c_j\}$ are predefined weights (based on $\{\lambda_j\}$) for each direction, and $\xi$ is adjusted to achieve proper coverage. We then have
\begin{align}
\label{e:EllipsoidBound}
E_{\hat{\theta}} = \left\{ h \in \mathcal{H} : \sum_{j=1}^{\infty} \frac{\langle \sqrt{N}(\hat{\theta} - h), v_j \rangle^2}{c_j^2} \leq \xi \right\}.
\end{align}
From \eqref{e:KL} it follows that the coverage is given by
\begin{align}
\label{e:EllipsoidProbRule}
\mbP \left(\theta \in E_{\hat{\theta}} \right) 
= \mbP \left( W_\theta \leq \xi \right)
\qquad
\text{where}
\qquad 
W_{\theta} = \sum_{j=1}^{\infty} \frac{\lambda_j}{c_j^2}Z_j^2.
\end{align}
Therefore, to achieve the desired asymptotic confidence level for a given $\{c_j\}$, one can take $\xi$ to be the $1-\alpha$ quantile of a weighted sum of chi-squared random variables. 
Though the distribution of the weighted sum of chi-squares does not have a closed form expression, fast and efficient numerical approximations exist such as the {\tt imhof} function in {\tt R} \citep{Imhof:1961:CDQ}. 

In choosing $\{c_j\}$ we suggest two important considerations.  The first is that one wants $c_j \to 0$ so as to eliminate the effect of later dimensions.  In doing so, one is also producing compact regions \citep{laha:roghatgi:1979}.  Since probability measures over Hilbert spaces are necessarily tight \citep{billingsley:1995}, meaning they concentrate on compact sets, a region which is not compact is overly large.  Conversely, the faster that $c_j \to 0$, the larger the mean of  $W_\theta$, which increases all of the radii.  Therefore, it seems desirable to balance these two concerns, choosing $c_j$ which go to zero, but not overly fast.   
\begin{figure}
	\centering
	\makebox{\includegraphics[width=\textwidth]{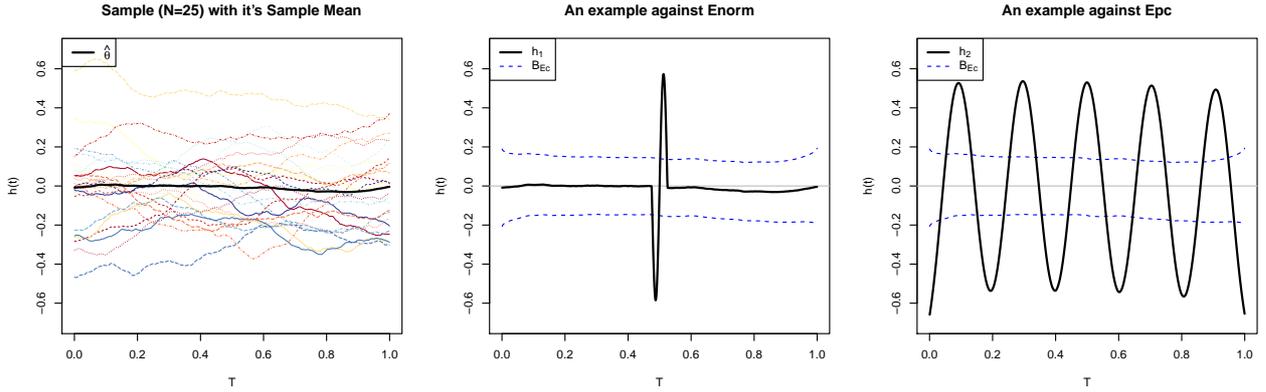}}
	\caption{\label{fig:ShortfallOfNoncompactRegion}For this illustration purpose, we take $\mcH$ = $L^2[0,1]$. The left plot shows an $iid$ mean zero Gaussian sample on $[0,1]$ along with its sample mean. Note that 95\% confidence region $E_{norm}$ contains all functions in $\mcH$ as long as they are close to the sample mean in $\mcH$ norm, in this case $L^2$ norm, like $h_1$ in the second column. However, a properly tailored confidence region $E_{c}$ yields a band form $B_{E_c}$ such that $E_{c} \subset B_{E_c}$, and $B_{E_c}$ (and therefore  $E_{c}$) exclude functions that are not bounded in the band \textit{almost everywhere}. Third column shows more extreme example of unbounded confidence regions. For the region $E_{PC(10)}^\circ$ (an \textit{opened-up} version), $h_2 := \hat{\theta} + b  v_{11}$ ($b$ being any real number) is inside the $E_{PC(10)}^\circ$ regardless of the confidence level.}
\end{figure}

Two popular hypothesis testing frameworks in FDA, the {\it norm approach} and {\it PC approach}, can be understood as two extreme cases in this framework. The norm approach to test $H_0: \theta = \theta_0$ uses $N\|\hat{\theta} - \theta_0 \|^2$ as the test statistic. If $H_0$ is true, this test statistic, asymptotically, is a weighted sum of $\chi^2_1$ random variables with weights $\{\lambda_j\}$. This corresponds to taking $c_j^2 = 1$ for all $j$. The resulting confidence region, which we denote as $E_{norm}$, is a ball in $\mcH$ and provides proper coverage for $\theta$. However, this region is not compact and therefore too large as illustrated in Figure \ref{fig:ShortfallOfNoncompactRegion}. 
The PC approach to hypothesis testing uses $\sum_{j=1}^{J} { N \langle \hat{\theta} - \theta_0, v_j \rangle^2}{\lambda_j^{-1}}$ as the test statistic, for some finite $J$. If $H_0$ is true, this test statistic follows a $\chi^2_J$ distribution; therefore, $J$ must be a finite value even when the covariance is known. There are two possible confidence regions induced by this approach.  Both regions take $c_j^2=\lambda_j$ for $j \leq J$, but for $j > J$, one could either \textit{close them off}, $c_j = 0 $, or \textit{open them up}, $c_j = \infty$. The former results in a compact confidence region, but we have $\mbP \left(\theta \in E_{\hat{\theta}} \right) = 0$, i.e. \textit{zero-coverage} even if we make the very artificial assumption that $\theta \in \Span \{v_1,\dots, v_J\}$ since the center of the region $\hat{\theta}$ sits outside $\Span \{v_1,\dots, v_J\}$ \textit{almost surely}. 
On the other hand, the \textit{opened-up} region would achieve proper coverage, but the region is not even bounded, let alone compact. 

There exists infinitely many options for proper $\{c_j\}$ and how to best choose them is an open question deserving further exploration. In preparing this work, a number of options 
were initially considered, however, we propose using the following due to 1) its ability to achieve \textit{the narrowest average squared width band} using tools from Section \ref{s:bands} 2) excellent empirical performance, and 3) its simplicity: 
\[
c_j^2 = \lambda_j^{1/2}
\qquad \text{and} \qquad E_c := \left\{ h \in \mathcal{H} : \sum_{j=1}^{\infty} \frac{\langle \sqrt{N}(\hat{\theta} - h), v_j \rangle^2}{\sqrt{\lambda_j}} \leq \xi \right\},
\] 
i.e. the square root of the corresponding eigenvalues. Although $\sum_j {\lambda_j}{c^{-2}_j} \equiv \sum_j {\lambda_j^{1/2}} < \infty$ is not always guaranteed, this holds for most processes that are smoother than Brownian motion ($\lambda_j \approx j^{-2}$), and therefore would hold in most applications. 
If the process is rough enough such that $\sum_j {\lambda_j^{1/2}} < \infty$ is not guaranteed, one may use another criteria suggested in the Appendices,  
namely $c_j^2 = \left(\sum_{i\geq j} \lambda_i \right)^{1/2}$, which guarantees both $c_j \to 0$ and $\sum_j {\lambda_j}{c^{-2}_j} < \infty$ \citep[p. 80]{rudin:1976}.

\subsection{Hyper-Rectangular Form}
\label{s:rectangle}
Our second form is a slight modification of the previous form, switching from an ellipse to a rectangle.    
In multivariate statistics a rectangular confidence region is often easier to interpret than an ellipse since it gives clear confidence intervals for each (principal component) coordinate.  However, it is often much easier to compute an ellipse since the distributions of quadratic forms are well understood. Regardless, we will show that it can still be easily computed using a nearly closed form expression, up to a function involving the standard normal quantile function. 

A hyper-rectangular region can be similarly constructed as:
\[
R_{\hat{\theta}} 
= \left\{ h \in \mathcal{H} : \frac{|  \langle h - \hat{\theta}, v_j \rangle |}{r_j} \leq 1, \forall \ j =1,2 \dots \  \right\} =  \left\{ h \in \mathcal{H} : |  \langle \sqrt{N}(h - \hat{\theta}), v_j \rangle  | \leq c_j\sqrt{\xi} , \ \forall j \right\} 
\]
using the same decomposition $r_j^2 = N^{-1} \xi c_j^2$. 
From the KL expansion \eqref{e:KL}, we want
\begin{align}
\label{e:RectangleProbRule}
\mbP \left(\theta \in R_{\hat{\theta}} \right) 
= \mbP \left( |\sqrt{\lambda_j}Z_j| \leq c_j \sqrt{\xi} , \ \forall j \right)  
\to 1-\alpha.
\end{align}
When we define $z_j := \frac{c_j}{\sqrt{\lambda_j} } \sqrt{\xi}$, the remaining problem is to find proper $\{z_j\}$, or selecting the $\{c_j\}$ and finding the proper $\xi$. One may first determine $\{c_j\}$ and find the proper $\xi$, or find the proper $\{z_j\}$ directly. Again, there exists infinitely many criteria and some examples can be found in the Appendices. Among those, we propose using the following: 
\[
z_j = \Phi_{sym}^{-1}\left[ \exp \left( {\frac{\lambda_j}{\sum_{k=1}^\infty \lambda_k} \log(1-\alpha)} \right) \right] \text{ for each } j,
\]
where $\Phi^{-1}_{sym}(\cdot)$ is defined as the inverse of $\Phi_{sym}(z):=\mbP(|Z| \leq z)$, $Z \stackrel{d}{=} \mcN(0,1)$.  We denote this rectangular region as $R_{z}$. This criterion produces a region that is close to the one that minimizes $\sup \{ \| h - \hat{\theta} \|^2 : h \in R_{\hat{\theta}} \}$, i.e. the distance between the farthest point of the region from the center,
but in a much faster way. It is simple, easy to compute, and shows an excellent empirical performance.

\subsection{Visualizing Ellipses via Bands}
\label{s:bands}
Visualizing a confidence ellipse is challenging even in the finite dimensional setting; it is very difficult once one goes beyond two or three dimensions.  In this sense, the rectangular regions are much easier to visualize since one can simply translate them into marginal intervals and examine each coordinate separately (while still achieving simultaneous coverage).    
It is therefore useful to develop visualization techniques for elliptical regions. One option is to construct bands in the form of an infinite collection of point-wise intervals over the domain of the functions.
To make our discussion more concrete, in this section only we assume that $\mcH = L^2(\mcD)$, where $\mcD$ is some compact subset of $\mbR^d$.   
For example, $d = 1$ for temporal curves and $d=2$ for spatial surfaces. 

A symmetric confidence band in $\mcH$ around $\hat{\theta}$ can be understood as:
\begin{align}
\label{e:bandform}
B_{\hat{\theta}} 
&= \left\{ h \in \mathcal{H} : |  h(x) - \hat{\theta}(x) | \leq r(x), \ \text{for } x \in \mcD \text{\textit{ almost everywhere}} \right\}.
\end{align}
The caveat ``almost everywhere" here (i.e. except on a set of Lebesgue measure zero) cannot be dropped since we are working with $L^2$ functions.  The downside of using the above band, however, is that an analytic expression for $r \in \mcH$ usually does not exist. One therefore typically resorts to simulation based methods as in \citet{degras:2011}. 
The band suggested by \citet{degras:2011} takes $c_\alpha \hat{\sigma}(t)/\sqrt{N}$ as $r(t)$ where $\hat{\sigma}(t)$ is the estimated standard deviation of $\sqrt{N}\hat{\theta}(t)$. The proper scaling factor $c_\alpha$ is then found via parametric bootstrap.  
We denote this band as $\hat{B}_{s}$, while denoting the one using the true covariance as $B_{s}$.

In traditional multivariate statistics and linear regression, ellipses can be transformed into point-wise intervals and bands using Scheff\'e's method which, at its heart, is an application of the Cauchy-Schwarz inequality.  This approach cannot be applied \textit{as is} to our ellipses because they are infinite dimensional.  However, a careful modification of Scheff\'e's method can be used to generate bands.  
We now show a $1-\alpha$ ellipsoid confidence region $E_{\hat{\theta}}$ can be transformed into the confidence band $B_{\hat{\theta}}$ such that $E_{\hat{\theta}} \subset B_{\hat{\theta}}$ based on a modification of Scheff\'e's method.
Defining 
\begin{align}
\label{e:BandZ1}
r(x) =  \sqrt{ \frac{\xi}{N}\sum_{j=1}^\infty c_j^2  v_j^2(x)}\ ,
\end{align}
then we have the following theorem.

\begin{theorem}\label{t:band}
	If Assumption \ref{a:normal} holds, $\sum c_j^2 < \infty$, and $\sum \lambda_j c_j^{-2}<\infty$, then $r(x) \in \mcH$ and $E_{\hat{\theta}} \subset B_{\hat{\theta}}$. Therefore, $\mbP(\theta \in B_{\hat{\theta}}) \geq 1-\alpha + o(1).$
\end{theorem}

These bands also lead to a convenient metric for choosing an ``optimal" sequence $c_j$.  In particular, we choose the $c_j$ which lead to a band with the \textit{narrowest average squared width}.  This is in general a difficult metric to quantify due to $\xi$.  However, we can replace $\xi$, which is a quantile of a random variable, by its mean to obtain the following:
\begin{align*}
ASW(\{ c_j \} ) = \sum_{j=1}^\infty \frac{\lambda_j}{c_j^2} \sum_{i=1}^\infty {c_i^2}.
\end{align*}
Clearly the $\{c_j\}$ are unique only up to a constant multiple, however, it is a straightforward calculus exercise to show that one option is to take $c_j^2 = \lambda_j^{1/2}$, which is also conceptually very simple.  It is also worth noting that this choice does not change with the smoothness of the underlying parameters or the covariance of the estimator; these quantities are implicitly captured by the eigenvalues themselves and thus already built into the $c_j$ with this choice.  
In practice, the coverage of this band will be larger than $1-\alpha$, since $E_{\hat{\theta}} \subset B_{\hat{\theta}}$ and the coverage of $E_{\hat{\theta}}$ is $1-\alpha$. Our simulation studies show that this gap is non-trivial for rougher processes, but narrows substantially for smoother ones.  The band formed this way from $E_c$ will be denoted as $B_{E_c}$. 

Our suggested band \eqref{e:BandZ1} takes into account the covariance structure of the estimator via the eigenvalues (though the $c_j$) and the eigenfunctions.  Thus, our band differs from those described in \cite{degras:2011} in that we do not use a constant threshold after taking into account the point-wise variance; our band adjusts locally to the within curve dependence of the estimator.  We will illustrate this point further in Section \ref{section:Simulation} as one of our simulation scenarios will have a dependence structure which changes across the domain.  Our band will adjust to this dependence, widening in areas with low within curve dependence and narrowing when this dependence is high. 

Lastly, one practical issue arises in finding proper $\xi$ since finding the quantile of weighted sum of $\chi^2$ random variables is not straightforward.  One may try to invert the approximate CDF, like \texttt{imhof} in $\texttt{R}$.  Alternatively, one can use a gamma approximation by matching the first two moments \citep{feiveson:delaney:1968}.   Our simulations showed that for typical choices of $\alpha$, such as $0.1, 0.05,$ or $0.01$, a gamma approximation works well.

\section{Estimating Confidence Regions and Ghosting}
\label{section:EstConfRegion}
We have, until now, treated $C_\theta$ as known for ease of exposition and to explore the infinite dimensional nature of the regions. In this section we consider the fully estimated versions.  Issues arise here that do not in the multivariate setting. In particular, one typically has \textit{zero-coverage} when working with estimated regions, but we will show that these regions are still in fact useful since they are very close to regions with proper coverage. In this sense, we call them \textit{Ghost Regions} since they `ghost' the regions with proper coverage.  Here we view the empirical regions as estimators of the desired regions which have proper coverage, and then show that the distance between the two quickly converges to zero. Our purpose in doing so is to provide a theoretical justification for using the regions in practice. In Section \ref{section:Simulation} we will also validate these regions through simulations. 

We assume that we have an estimator $\hat{C}_\theta$ of $C_\theta$ which achieves root-$N$ consistency.  
Consistency of $\hat{C}_\theta$ enables us to replace $\{( v_j , \lambda_j) \}_{j=1}^{\infty}$ with the empirical versions $\{( \hat{v}_j , \hat{\lambda}_j) \}_{j=1}^{N}$\footnote{In practice we usually have less than $N$ empirical eigenfunctions due to the estimation of other parameters.}. 
When we replace $\{v_j\}_{j=1}^\infty$ with $\{\hat{v}_j\}_{j=1}^{N}$, however, we nearly always end up with a finite number of estimated eigenfunctions (with nonzero eigenvalues). 
We present asymptotic theory for the hyper-ellipsoid form although similar arguments can be applied to the hyper-rectangular form.

Define $\cH_J := \Span ( \{\hat{v}_j\}_{j=1}^{J} ) \subset \mcH$ where $J \leq N$. We construct two versions of the estimated confidence regions 
\begin{align}
\label{e:estimatedregion-open}
\hat{E}^{\circ}_{\hat{\theta}} = \left\{ h \in \cH : \sum_{j=1}^{J} \frac{\langle h - \hat{\theta}, \hat{v}_j \rangle^2}{N^{-1}c_j^2} \leq \xi \right\} \qquad \text{and}
\end{align}
\begin{align}
\begin{split}
\label{e:estimatedregion}
\hat{E}_{\hat{\theta}} 
= \left\{ h \in \cH_J : \sum_{j=1}^{J} \frac{\langle h - \hat{\theta}, \hat{v}_j \rangle^2}{N^{-1} c_j^2} \leq \xi \right\} 
= \left\{ h \in \cH: \sum_{j=1}^{\infty} \frac{\langle h - \hat{\theta}, \hat{v}_j \rangle^2}{N^{-1} c_j^2 \mathbf{1}_{j \leq J}} \leq \xi \right\},
\end{split}
\end{align}
though in our theoretical results we will let $J \to \infty $ with $N$.  The empirical eigenfunctions $\{\hat{v}_j\}_{j=1}^J$ can be extended to give a full orthonormal basis of $\cH$. Note that $\hat{E}_{\hat{\theta}}$ is `closed off' while $\hat{E}^{\circ}_{\hat{\theta}}$ is `opened up' for those dimensions not captured by the first $J$ components.
We take $\xi$ to be the $1-\alpha$ quantile of a weighted sum of $\chi^2$ random variables with weights $\{{\hat \lambda_j}{c^{-2}_j} \}_{j=1}^J$.  
Observe that $\hat{E}_{\hat{\theta}}^{\circ}$ achieves the proper coverage $\mbP ( \theta \in \hat{E}_{\hat{\theta}}^{\circ} ) \to 1-\alpha$. However, $\hat{E}_{\hat{\theta}}^{\circ}$ cannot be compact  regardless of how $\{c_j\}$ is chosen unless $\mcH$ is finite dimensional. If we quantify the distance between sets using Hausdorff distance, $\hat{E}_{\hat{\theta}}^{\circ}$ does not converge to $E_{\hat{\theta}}$ since it is unbounded. 
On the other hand, $\hat{E}_{\hat{\theta}}$ is always compact but has \textit{zero-coverage}; we  almost always have $\mbP ( \theta \in \hat{E}_{\hat{\theta}} ) = 0$ regardless of the sample size. 
Therefore, neither empirical confidence regions maintains the nice properties of the ones using a known covariance -- compactness and proper coverage -- at the same time. However, as we will show, $\hat{E}_{\hat{\theta}}$ is close to $E_{\hat{\theta}}$ in Hausdorff distance, meaning we can use $\hat{E}_{\hat{\theta}}$ as an estimate of the desired region ${E}_{\hat{\theta}}$. With this convergence result at hand, one may prefer the closed version $\hat{E}_{\hat{\theta}}$ over $\hat{E}_{\hat{\theta}}^{\circ}$ as a confidence region. Because $\hat{E}_{\hat{\theta}}$ does not have proper coverage we call it a \textit{ghost} region.

\subsection{Convergence in the Hausdorff Metric}
In this subsection we show that the Hausdorff distance, denoted $d_H$, between $\hat{E}_{\hat{\theta}}$  and ${E}_{\hat{\theta}}$ can be well controlled. In particular, we will show that this distance converges to zero faster than $N^{-1/2}$.  Since this is the rate at which ${E}_{\hat{\theta}}$ shrinks to a point, this is necessary to ensure that $\hat{E}_{\hat{\theta}}$ is actually useful as a proxy for ${E}_{\hat{\theta}}$.
We begin by introducing a fairly weak assumption on the distribution of $\hat{C}_\theta$.  Recall that $C_\theta$ is a Hilbert-Schmidt operator (all covariance operators are) in the sense that 
$
\|C_{\theta}\|^2_{\mcS} := \sum_{j=1}^\infty \| C_{\theta}(e_j) \|^2_\mcH < \infty
$
where $\{e_j\}$ is any orthonormal basis of $\mcH$. We denote the vector space of Hilbert-Schmidt operators by $\mcS$, which is also a real separable Hilbert space with inner product
$
\langle \Psi, \Phi \rangle_{\mcS} := \sum_{j=1}^\infty \langle \Psi(e_j), \Phi(e_j)  \rangle_{\mcH}.
$
A larger space, $\mcL$, consists of all bounded linear operators with norm
$
\| \Psi\|_{\mcL} = \sup_{h \in \mcH}  \| \Psi (h)\|  / \|h\|,
$
which is strictly smaller than the $\mcS$ norm, implying $\mcS \subset \mcL$.  
We now assume that we have a consistent estimate of $C_\theta$.
\begin{assumption} \label{assumption:1}
	Assume that we have an estimator $\hat C_\theta$ of $C_\theta$ which is root-$N$ consistent in the sense that $N  \mbE \| \hat C_\theta - C_\theta\|_{\mcS}^2 = O(1)$.
\end{assumption}
The Hausdorff distance between two subsets $S_1$ and $S_2$ of $\cH$ is defined as
\[
d_H(S_1,S_2) = \max\{\rho(S_1,S_2), \rho(S_2,S_1) \}, \qquad \text{where} \qquad \rho(S_1,S_2) = \sup_{x \in S_1} \inf_{y \in S_2} \|x - y \|_\cH.
\]
We say two regions $S_1$ and $S_2$ converge to each other if $d_H(S_1,S_2)$ converges to $0$. Therefore, to achieve convergence of $\hat{E}_{\hat{\theta}}$ to $E_{\hat{\theta}}$ in probability, we need $d_H(\hat{E}_{\hat{\theta}},E_{\hat{\theta}}) \xrightarrow{P} 0 $ as $N \to \infty$.
To accomplish this, we separate the results for $\rho(\hat{E}_{\hat{\theta}},E_{\hat{\theta}})$ and $\rho({E}_{\hat{\theta}},\hat E_{\hat{\theta}})$.  Since $\hat{E}_{\hat{\theta}}$ is the ``smaller" set, the former is primarily controlled by the distance between the empirical and population level eigenfunctions. The latter is additionally influenced by how large the remaining dimension of $E_{\hat{\theta}}$ is. We define $\{\alpha_j\}$ as
\[
\alpha_1 := \lambda_1 - \lambda_2 \quad \text{and} \quad \alpha_j := \min\{ \lambda_j - \lambda_{j+1}, \lambda_{j-1} - \lambda_{j} \} \text{ for } j = 2, \dots.
\]
Our primary convergence results are given in the following two theorems.
\begin{theorem} 
	\label{thm:Convergence1}
	If Assumptions \ref{a:normal} and \ref{assumption:1} hold and $c_1 \geq c_2 \geq \dots$ then with probability one
	\begin{align}
	\label{e:order1}
	\rho(\hat{E}_{\hat{\theta}},E_{\hat{\theta}}) \leq \left[ \sum_{j=1}^{J}  \frac{8 \xi  c_1^2 \|\hat C_\theta - C_\theta\|_\mcL^2}{N \alpha_j^{2}} \right]^{\frac{1}{2}}.
	\end{align}
\end{theorem}

\begin{theorem} 
	\label{thm:Convergence2}
	If Assumptions \ref{a:normal} and \ref{assumption:1} hold and $c_1 \geq c_2 \geq \dots$ then with probability one
	\begin{align}
	\rho(E_{\hat{\theta}}, \hat{E}_{\hat{\theta}})  
	\leq  \left[c_J^2  N^{-1} \xi \right]^{\frac{1}{2}}  + \left[ \sum_{j=1}^{J}  \frac{8 \xi  c_1^2 \|\hat C_\theta - C_\theta\|_\mcL^2}{N \alpha_j^{2}} \right]^{\frac{1}{2}}.
	\label{e:order2}
	\end{align}
\end{theorem}

With Theorems \ref{thm:Convergence1} and \ref{thm:Convergence2} in hand, we can characterize the overall convergence rate for $d_H(\hat{E}_{\hat{\theta}},E_{\hat{\theta}}) $, but we first need more explicit assumptions on the rates for the eigenvalues, $\lambda_j$, and weights, $c_j^2$.  
\begin{assumption}
	\label{a:rate1}
	Assume that there exist constants $K > 1 $, $\delta > 1$, and $\gamma > 0$ such that  
	$$\frac{1}{Kj^{\delta}}  \leq \lambda_j \leq \frac{K}{  j^{\delta}},
	\qquad
	\frac{1}{Kj^{\delta+1}}  \leq \lambda_j - \lambda_{j+1} \leq \frac{K}{ j^{\delta+1}},
	\quad  \text{ and } \quad
	\frac{1}{K j^{ 2 \gamma}} \leq c_j^2 \leq \frac{K}{j^{ 2 \gamma}},
	$$ 
	for all $j = 1, \dots$, where we have $0 < 2 \gamma < \delta -1$.
\end{assumption}
The first two assumptions are quite common in FDA.  One needs to control the rate at which the eigenvalues go to zero as well as the spread of the eigenvalues which influences how well one can estimate the corresponding eigenfunctions, though this can likely be slightly relaxed \citep{reimherr:2015}.  The rate at which $c_j^2$ decreases to zero also needs to be well controlled, and, in particular, it cannot go to zero much faster than $\lambda_j$.

\begin{theorem}
	\label{thm:ConvergenceOverall}
	Assume that Assumptions \ref{a:normal},  \ref{assumption:1} and \ref{a:rate1} hold, then
	\begin{enumerate}
		\item The $J$ which balances \eqref{e:order1} and \eqref{e:order2} is $J = N^{\frac{1}{2 \delta + 3 + 2 \gamma}}$.
        \item The overall convergence rate is then
		$\mbE [d_H(\hat{E}_{\hat{\theta}},E_{\hat{\theta}})^2 ] \leq   
		O \left( N^{-\left( 2 - \frac{2\delta + 3}{2 \delta + 3 + 2\gamma} \right)} \right)$.
	\end{enumerate}
\end{theorem}

Theorem \ref{thm:ConvergenceOverall} shows that the squared distance between the ghost region $\hat{E}_{\hat{\theta}}$ and the desired region $E_{\hat{\theta}}$ goes to zero faster than $N^{-1}$, which is the rate at which $E_{\hat{\theta}}$ shrinks to a point.  This suggests that $\hat{E}_{\hat{\theta}}$ is a viable proxy for ${E}_{\hat{\theta}}$, even though it has zero coverage. 

Interestingly, the rate is better the faster that $c_j$ tends to zero, i.e. for larger values of $\gamma$.  At first glance, this may suggest that one should actively try and find $c_j$ which tend to zero as fast as possible. However, by changing the $c_j$ one is changing the confidence region into a potentially less desirable one.  
In particular, as we will soon see the choice of $c_j^2 = \lambda_j^{1/2}$ leads to a suboptimal convergence rate in Theorem \ref{thm:ConvergenceOverall}, but, in some sense, leads to an optimal confidence band and excellent empirical performance.  Thus, this may be one of the few instances in statistics where it is not necessarily desirable to have the ``fastest" rate of convergence. 

We finish this section by stating a Corollary for when $c_j^2  = \hat \lambda_j^{1/2}$.  In this case, we also take into account that the $c_j$ are estimated from the data.  Note that in Theorem \ref{thm:ConvergenceOverall} it is assumed that the $c_j$ are not random, while in Theorems \ref{thm:Convergence1} and \ref{thm:Convergence2} the $c_j$ can be random or deterministic as long as they are nonincreasing.
\begin{corol}
	\label{corol:EchXConvRate1} Let $c_j^2 = \hat \lambda_j^{1/2}.$  Then under Assumptions \ref{assumption:1} and \ref{a:rate1},   
	we have 
	$$d(\hat E_{\hat \theta}, E_{\hat \theta})^2 =   O_p\left(N^{ - \frac{6\delta + 6}{5 \delta + 6}} \right).$$
\end{corol}

\section{Simulation}
\label{section:Simulation}

In this section, we present a simulation study to evaluate and illustrate the proposed confidence regions and bands. Throughout this section, we only consider dense FDA designs. Section \ref{subsection:HT} first compares different regions for hypothesis testing. Note that comparing these regions presents a nontrivial challenge as we cannot just choose the ``smallest" one as we are working in infinite dimensional Hilbert spaces.  We therefore turn to using the regions for hypothesis testing, evaluating each's ability to detect different types of changes from some prespecified patterns.
In Section \ref{subsection:ComparisonOfBands}, we visually compare bands and examine their local coverages. Lastly, in Section \ref{subsection:SimulationOnData}, we consider more complicated mean and covariance structures borrowed from the DTI data in Section \ref{s:dti} and examine the effects of smoothing.

\subsection{Hypothesis Testing}
\label{subsection:HT}

We consider the hypothesis testing $H_0 : \theta =  \theta_0$ vs. $H_1 : \theta \neq \theta_0 $. For a given confidence region $G_{\hat{\theta}}$, the natural testing rule is to reject $H_0$ if $\theta_0 \notin G_{\hat{\theta}}$. 
For ellipses and rectangles, however, we compare $\theta_0$ only in the directions included in the construction of the confidence regions to alleviate the ghosting issue and mimic how the methods would be used in practice. 
We calculate p-values (detailed in the Appendices) and compare them to $\alpha$. For bands like $B_s$ and $B_{E_c}$, $H_0$ will be rejected if $\theta_0$ sits outside the bands at least one evaluation point over the domain.

In this section and in Section \ref{subsection:ComparisonOfBands}, we take $\mcH = L^2[0,1]$ and consider an $iid$ sample $\{X_i(t)\}_{i=1}^N$, $t\in[0,1]$ from a Gaussian process $\mcN(\theta, C_\theta)$. 
To estimate $\theta$ and $C_\theta$ (when unknown), we use the standard estimates \citep{hkbook}  $\hat{\theta}(t) = N^{-1}\sum_{i=1}^{N}X_i(t)$ and 
$\hat{C}_\theta(t,s) = (N-1)^{-1} \sum_{i=1}^N  ( X_i(t) - \hat{\theta}(t) ) ( X_i(s) - \hat{\theta}(s) )$. 
To emulate functions on the continuous domain $[0,1]$, functions are evaluated at 100 equally spaced points over $[0,1]$.

\subsubsection{Verifying Type I Error}
\label{subsection:TypeIError}

\paragraph{Regions with Known Covariance:} We first verify Type I error rates assuming the true covariance operator is known.  Each setting was repeated 50,000 times according to the following procedure: 

\begin{enumerate}
	\item Generate a sample $\{X_i\}_{i=1}^N \iid \mcN(\theta,C_\theta)$. For the mean function we take $\theta(t) := 10t^3-15t^4+6t^5$, which was used in \citet{degras:2011} and \citet{Hart:1986:KRE}.  For the covariance operator, we use a Mat\'ern covariance   
		$C_\theta(t,s) :=  \frac{.25^2}{\Gamma(\nu) 2^{\nu-1}} \left( \sqrt{2\nu}|t-s|\right)^\nu K_\nu  \left(\sqrt{2\nu}|t-s| \right)$, where  $K_\nu(\cdot)$ is the modified Bessel function of the second kind, and $\nu$ is the smoothness parameter.
	\item Find $\hat{\theta}$ from the sample, and perform hypothesis testings on $\theta_0 = 10t^3-15t^4+6t^5$, which is the same as $\theta$, based on the confidence regions using $C_\theta$, the true covariance. 
\end{enumerate}
To represent small/large sample size and rough/smooth processes, the four combinations of $N=25$, $N=100$ and $\nu=1/2,$ $\nu=3/2$ were used. Table \ref{tbl:ActualAlphaTrueCov} summarizes proportion of the rejections. 

\begin{table}[h]
	\centering
	\caption{Type I error rates with known covariance. $E_{norm}$, $E_{PC}$, and $E_{c}$ represent ellipsoid regions from norm approach, FPCA approach, and the proposed one, respectively. $B_{s}$ is the simulation based band while $B_{E_c}$ is the band based on $E_{c}$. $R_{z}$ is the proposed rectangular region and $R_{zs}$ is the small sample version of $R_{z}$, which uses only eigenfunctions (but not eigenvalues) of $C_\theta$.}
	\label{tbl:ActualAlphaTrueCov}
	\begin{tabular}{rl|rrr|rrrr}
		\hline 
		$N$   & $\nu$        & $E_{norm}$      & $E_{PC}$ & $B_{s}$  & 
		$E_{c}$ & $R_{z}$ & $R_{zs}$ &
		$B_{E_c}$  \\ \hline
		25  & $\sfrac{1}{2} \ (rough)$ & .048 & .049 & .049 & .049 & .048 & .049 & .000  \\
		25  & $\sfrac{3}{2} \ (smooth)$ & .051 & .051 & .053 & .050 & .051 & .049 & .025  \\
		100 & $\sfrac{1}{2} \ (rough)$ & .051 & .049 & .052 & .051 & .050 & .050 & .000  \\
		100 & $\sfrac{3}{2} \ (smooth)$ & .050 & .049 & .047 & .049 & .049 & .049 & .023 \\ \hline
	\end{tabular}
\end{table}

All the methods are satisfactory except for the transformed band $B_{E_c}$, which generates a conservative band as expected. For the ellipsoid and rectangular regions, up to the very last PCs were used -- trimming out only $\lambda_j < 10^{-18}$ -- and the results were still stable. Although not presented here, the results were robust against the number of PCs used.

\paragraph{Regions with Unknown Covariance:}

We now use $\hat C_\theta$ instead of $C_\theta$ in the step 2 above, 
 and use PCs to capture at least $99.9\%$ of estimated variance, i.e. took $J$ such that
$J = \min_j ( \sum_{i=1}^j \hat{\lambda}_i / \sum_{i=1}^{N-1}\hat{\lambda}_i \geq .999)$, for all ellipsoid and rectangular regions. For the FPCA based region, we  additionally took $J=3$, which explained approximately $90\%$ of the variability. Table  \ref{tbl:ActualAlpha} summarizes the proportions of the rejections and the following can be observed.
\begin{enumerate}
	\item Coverage of $\hat{E}_{PC}$ is very sensitive to the number of PCs used and works well only when the number is relatively small. This reenforces the common concern of how to best choose $J$ in practice. In contrast, $\hat{E}_{norm}$ does not have this question. Our proposed methods $\hat{E}_{c}$ and $\hat{R}_{z}$ lie somewhere between the two and choosing $J$ is not a concern as long as the very late PCs are dropped. 
	\item When $N$ is small, the small sample modification of the rectangular region ($\hat{R}_{zs}$) achieves slightly conservative but seemingly the best result. $\hat{E}_{norm}$ follows closely, possibly due to its lower dependency on later PCs. The details on $\hat{R}_{zs}$ can be found in the Appendices. 
\end{enumerate}
\begin{table}[h]
	\centering
	\caption{Type I error rates with an estimated covariance. $E_{norm}$, $E_{PC}$, and $E_{c}$ represent ellipsoid regions from norm approach, FPCA approach, and the proposed one, respectively. $B_{s}$ is the simulation based band while $B_{E_c}$ is the band based on $E_{c}$. $R_{z}$ is the proposed rectangular region and $R_{zs}$ is the small sample version of $R_{z}$, which uses only eigenfunctions (but not eigenvalues) of $C_\theta$.}
	\label{tbl:ActualAlpha}
	\begin{tabular}{rr|rrrr|rrrrrrr|r}
		\hline 
		$N$   & $\nu$  & $\hat{E}_{norm}$ & $\hat{E}_{PC}$ & $\hat{E}_{PC(3)}$ & $\hat{B}_{s}$ &
		$\hat{E}_{c}$   & $\hat{R}_{z}$ & $\hat{R}_{zs}$  &
		$\hat{B}_{E_c}$  &  $PC^*$   \\ \hline
		25  & $\sfrac{1}{2}$ & .057 & .162 & .069 & .087 & .071 & .069 & .041 & .013 &  21  \\
		25  & $\sfrac{3}{2}$ & .061 & .132 & .090 & .071 & .068 & .068 & .047 & .039 &  5  \\
		100 & $\sfrac{1}{2}$ & .052 & .255 & .056 & .058 & .060 & .059 & .050 & .001 &  53  \\
		100 & $\sfrac{3}{2}$ & .052 & .066 & .057 & .054 & .053 & .053 & .049 & .026 &  5  \\ \hline
	\end{tabular}
	\begin{center}
		\textit{* Median number of PCs required to capture $\geq 99.9\%$ of estimated variance.}
	\end{center} 
\end{table}
We emphasize here the dependence on choosing $J$ for both the FPCA and our new approach.  As is well known, FPCA based methods are very sensitive to the choice of $J$ as it places all eigenfunctions on an ``equal footing''.  However, later eigenfunctions are often estimated very poorly, which can result in very bad type 1 error rates when $J$ is taken too large.  In contrast, our approach is not as sensitive to the choice of $J$ since later eigenfunctions are down weighted.  In our simulations, they remained well calibrated as long as the very late FPCs are dropped, e.g. after capturing $99\%$ of the variance.

\subsubsection{Comparing Power}

To compare the power of the hypothesis tests, we gradually perturb $\theta$ -- the actual sample generating mean function -- from $\theta_0$ by an amount $\Delta \in \mbR$. To emulate what one might encounter in practice, three scenarios are considered: 
\begin{enumerate}
	\item shift: $\theta_0(t) = 10t^3-15t^4+6t^5$, \quad $\theta(t) = \theta_0(t) + \Delta$, 
	\item scale: $\theta_0(t) = 10t^3-15t^4+6t^5$, \quad $\theta(t) = \theta_0(t)(1+\Delta)$,
	\item local shift: $\theta_0(t) = \max\left\{0, -10|t-0.5| + 1 \right\}$, \quad $\theta(t) = \max\left\{0, -10|t-0.5| + 1+\Delta \right\}$. 
\end{enumerate}

A visual representation of the three scenarios can be found in the left column of Figure \ref{fig:PowerComparison}.  We estimate $C_\theta$ throughout, reduce the number of repetitions to 10,000, and use the same combinations of the sample size ($N=25, \ 100$) and smoothness ($\nu=$\sfrac{1}{2}$, \ $\sfrac{3}{2}$ $). For the $\hat{E}_{PC}$ method, the first 3 PCs were again used to ensure an acceptable Type I error. For other ellipsoid and rectangular regions, $J$ was taken to explain approximately  99.9\% of the variance as in the previous section. Power plots for $N=100$ and $\nu=\sfrac{1}{2}$ can be found in the right column of Figure \ref{fig:PowerComparison}, and a summary is given in Table \ref{tbl:AveragePower]}. The result for other combinations of sample size and smoothness can be found in the Appendices, but they all lead to the same conclusions:
\begin{enumerate}
\item In scenario 1, $\hat{E}_{PC(3)}$ has the lowest power while the other regions performs similarly. 
\item In scenario 2, $\hat{E}_{PC(3)}$ has the highest power while $E_{norm}$ has the lowest.  Our hyper-ellipse method $\hat E_c$ has only slightly less power than the FPCA method. Our rectangular method, $\hat R_z$, and the band of \citet{degras:2011} have about the same power, but both are lower than the ellipse. 
\item In Scenario 3, our proposed regions $\hat{E}_{c}$ and $\hat{R}_{z}$ far outperform existing ones. Note that $\theta$ differs from $\theta_0$ only on a fraction of the domain and the size of the departure is also small. Due to the small $\| \theta - \theta_0\|$, therefore, $\hat{E}_{norm}$ performs the worst. The FPCA method, $\hat{E}_{PC(3)}$ and Degras's band fall quite a bit behind our proposed methods, but still better than the norm approach. The $\hat{E}_{PC(3)}$ performs much better when the process is smooth and therefore the `signal' is captured in earlier dimensions -- although it still falls short from the proposed ones.       
\end{enumerate}

As a conclusion, we recommend using $\hat E_c$ in practice for hypothesis testing purposes. We base this recommendation on  1) its power is at the top or near the top in every scenario; 2) its type I error is well-maintained as long as very late PCs are dropped; 3) it is less sensitive to the number of PCs used as long as the number is reasonably large; 4) it is easy to compute; and 5) it can be used to construct a band. Being able to make this recommendation is quite substantial as previous work has focused on the norm versus PC approach, where clearly one does not always outperform the other \citep{ReNi:2014}.
\begin{table}
	\centering
	\caption{Average Power over $\Delta$ for each Scenario}
	\label{tbl:AveragePower]}
	\begin{tabular}{l|rrr|rrr}
		\hline 
		Scenario &  $\hat{E}_{norm}$ & $\hat{E}_{PC(3)}$ & $\hat{B}_{s}$ &
		$\hat{E}_{c}$   & $\hat{R}_{z}$ & $\hat{R}_{zs}$ \\ \hline
		1. Shift       &  .623 & .560 & .617 & .625 & .607 & .598  \\
		2. Scale       &  .411 & .549 & .503 & .522 & .496 & .480  \\
		3. Local Shift &  .234 & .568 & .504 & .759 & .770 & .749  \\ \hline
	\end{tabular}
\end{table}
\begin{figure}
	\centering
	\makebox{\includegraphics[width=\textwidth]{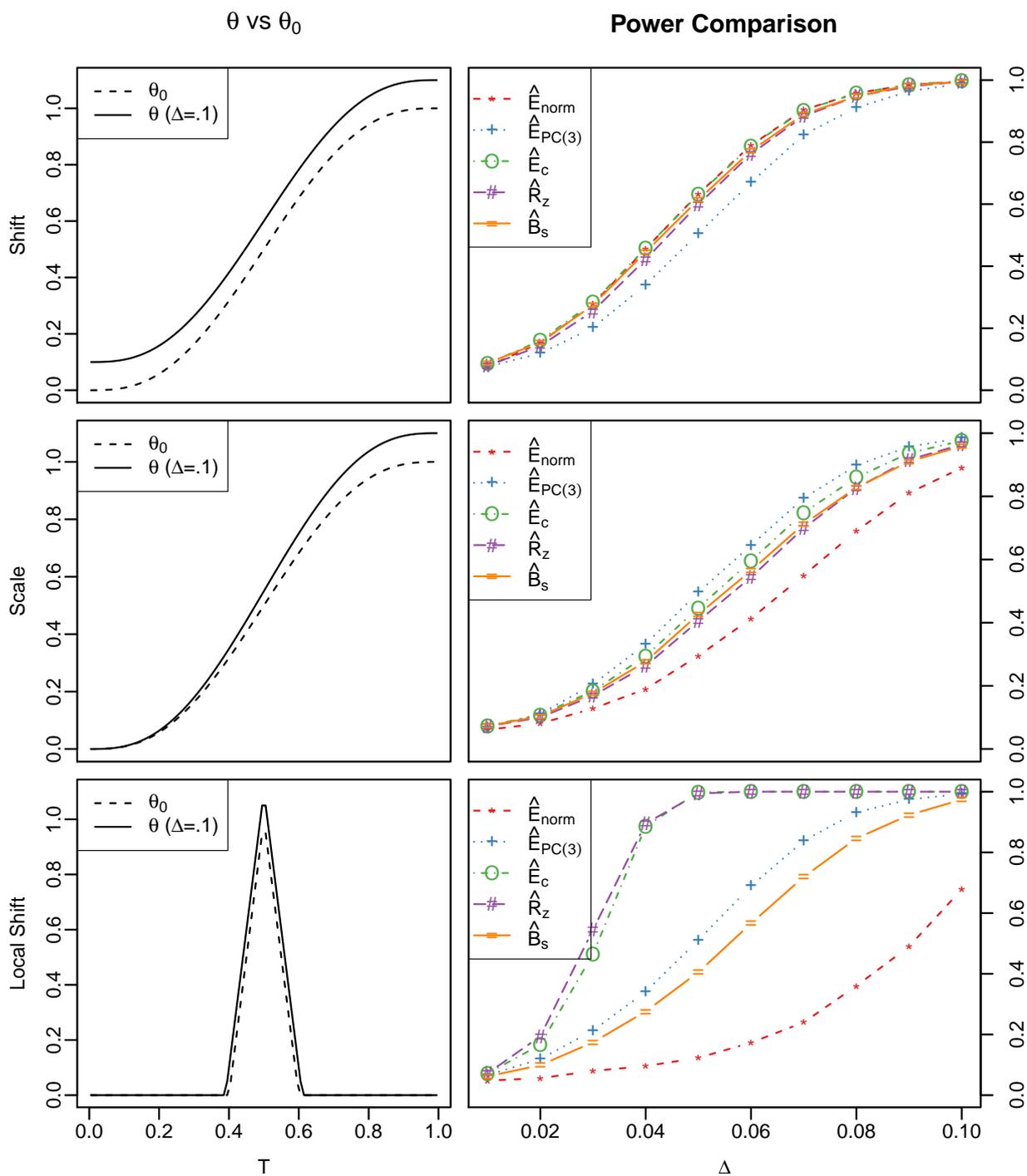}}
	\caption{\label{fig:PowerComparison}Power Comparison for each Scenario ($N=100, \ \nu=1/2$)}
\end{figure}

\subsection{Comparison of Bands}
\label{subsection:ComparisonOfBands}
In this section we compare the shape of $\hat{B}_{E_c}$ with $\hat{B}_{s}$, the two band forms of confidence regions, along with point-wise 95\% confidence intervals denoted as `naive-t'.   For this purpose, we consider three different scenarios regarding the smoothness of $\hat{\theta}$. The procedure can be summarized as follows:

\begin{enumerate}
	\item Generate a sample $\{X_i\}_{i=1}^N \iid \mcN(\theta,C_\theta)$, using the same mean function $\theta(t)$ as  in Section \ref{subsection:TypeIError}. For the covariance operator, three scenarios are considered:
	\begin{enumerate}
		\item The same Mat\'ern covariance in Subsection \ref{subsection:TypeIError} with $\nu = \sfrac{1}{2}$ (rough).
		\item The same Mat\'ern covariance in Subsection \ref{subsection:TypeIError} with $\nu = \sfrac{3}{2}$ (smooth).
		\item $C_\theta(t,s) :=  \frac{.25^2}{\Gamma(\nu) 2^{\nu-1}} \left( \sqrt{2\nu}|t^{10}-s^{10}|\right)^\nu K_\nu  \left(\sqrt{2\nu}|t^{10}-s^{10}| \right)$ with $\nu=\sfrac{1}{2}$. This generate processes that transition from smooth to rough by `warping' the domain. 
	\end{enumerate}
	\item Find $\hat{\theta}$ and $\hat{C}_\theta$, and generate symmetric bands around $\hat{\theta}$ using $\hat{C}_\theta$.
\end{enumerate}

Figure \ref{figure:ConfidenceBands} shows sample paths ($N=25$) from the three different covariance operators on the first row,  their 95\% simultaneous confidence bands on the second row, and local coverage rates on the third row. The findings can be summarized as follows :

\begin{enumerate}
	\item The proposed band $\hat{B}_{E_c}$ is wider than $\hat{B}_{s}$ for rougher processes, but almost identical to $\hat{B}_{s}$ for smoother ones, except for the far ends of the domain.
	\item In the third case, $\hat{B}_{E_c}$ is narrower than $\hat{B}_{s}$ in the smoother areas, while $\hat{B}_{s}$ maintains the same width.  $\hat{B}_{E_c}$ adjusts its width such that it gets narrow in the smooth areas (higher within curve dependence) and wider in the rough areas.
	\item Due to its construction, $\hat{B}_{E_c}$ does not bear any local under-coverage issue, and therefore any pattern in the third row in Figure \ref{figure:ConfidenceBands} can be rather related to its over-coverage. 
\end{enumerate} 

\begin{figure}[h]
	\centering
	\makebox{\includegraphics[width=\textwidth]{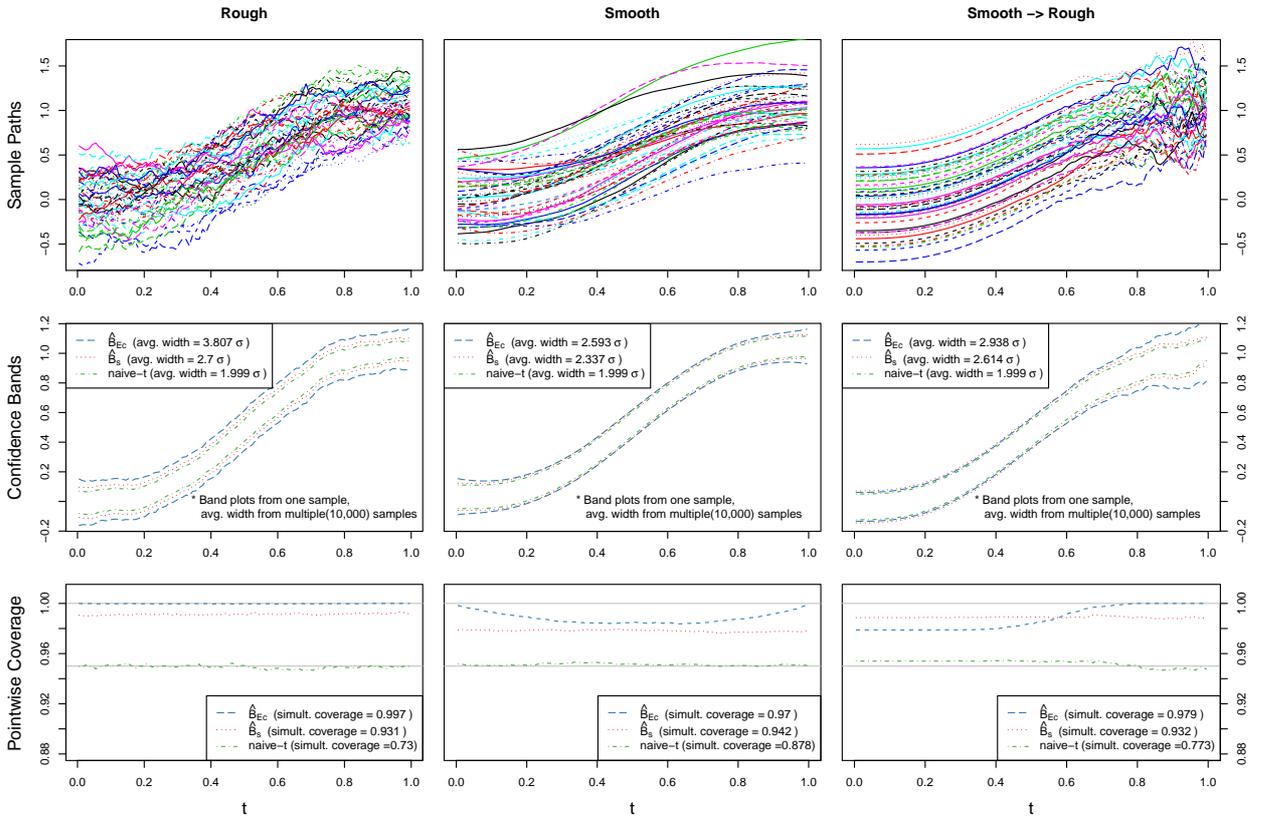}}
	\caption{\label{figure:ConfidenceBands} For each column, we have sample paths from a sample (1st row), $95\%$ confidence bands constructed from the sample (2nd row), and point-wise coverage rates from multiple (10,000) samples (3rd row). The `ave. width's in legends of 2nd row are averaged over the domain $[0,1]$ and over the samples, and was shown as multiple of point-wise (true) standard deviation.}
\end{figure}

We conclude that the confidence band $\hat{B}_{E_c}$ is an effective visualization tool to use in practice especially when the estimate $\hat{\theta}$ is relatively smooth. For smoother estimates, it is nearly identical to the parametric bootstrap but is much faster to compute since it requires no simulation. This is important as our band is conservative, utilizing a Scheff\'e-type inequality. It suggests that not much is lost in using such an approach as long as the parameter estimate is sufficiently smooth.  If the hypothesis tests and our confidence bands are in disagreement, say due to the conservative nature of $\hat{B}_{E_c}$, then it is recommended to follow up with a parametric bootstrap to get tighter bands.

\subsection{Simulation based on DTI data}
\label{subsection:SimulationOnData}
Although the mean and the covariances in Subsection \ref{subsection:HT} and \ref{subsection:ComparisonOfBands} are chosen to mimic common functional objects, actual data in practice may show much more complex structures. In this subsection we use a mean and a covariance structure from the \texttt{DTI} dataset in the \texttt{R} package \texttt{refund}. \textit{This DTI data were collected at Johns Hopkins University and the Kennedy-Krieger Institute}.  More details on this data set can be found in \citet{goldsmith:2012a} and \citet{goldsmith:2012b}.  This dataset contains fractional anisotropy tract profiles of the corpus callosum for two groups -- healthy control group and multiple sclerosis case group, observed over 93 locations. In this subsection, we took only the first visit scans of the case group in which the sample size is 99. We will denote this sample as original sample. 

First, we estimated the sample mean and covariance from the original sample and considered them as parameters. While the mean was estimated by penalizing $2^{nd}$ derivative with leave-one-out cross-validation to achieve a smooth mean function, the covariance was estimated using the standard method (but using the smoothed mean) to mimic the roughness of the original sample.
Using these mean and covariance, we generated multiple (10,000) Gaussian simulation samples of the sample size $99$.  
The use of Gaussian sample could be justified by the distribution of coefficients on each principal components in the original sample.

For each generated sample, two different estimation procedures were taken to look at the effect of smoothing. First approach is to simply smooth the sample using quardic bspline basis (with equally spaced knots) and use the standard estimates, and the second approach is to directly smooth the mean function by penalizing $2^{nd}$ derivative (or curvature), in which the covariance was estimated accordingly as shown in the Appendices. For both approaches, leave-one-out cross validation was used to find the number of basis functions and the penalty size, respectively. To reduce the computation time, those values were pre-determined from the original sample and applied to the simulation samples. 

Other than the explicit differences in approaches for  smoothing -- data first vs directly on the estimate, and bspline vs penalty on curvature --, the first smoothing would introduce bias because the 15 bspline functions would not fully recover the assumed mean function. The empirical bias from the first smoothing was $6.1$ times larger than the second one. 

Figure \ref{figure:PointwiseCoverage} compares confidence bands from the two smoothing schemes. Although they do not snow any material difference in the shapes of the bands, we get slightly narrower bands from bspline smoothing (left column). This may cause under-coverage for $\hat{B}_s$, but does not work adversely for the proposed band $\hat{B}_{E_c}$ which generally provides over-coverage. While the narrow band for $\hat{B}_{E_c}$ is mainly caused by more explicit dimension reduction or more smoothing, but for $\hat{B}_s$ and naive-t, bias seems to be the main source of it -- This can be supported by the local coverage patterns in the figure.
\begin{figure}[h]
	\centering
	\makebox{\includegraphics[width=\textwidth]{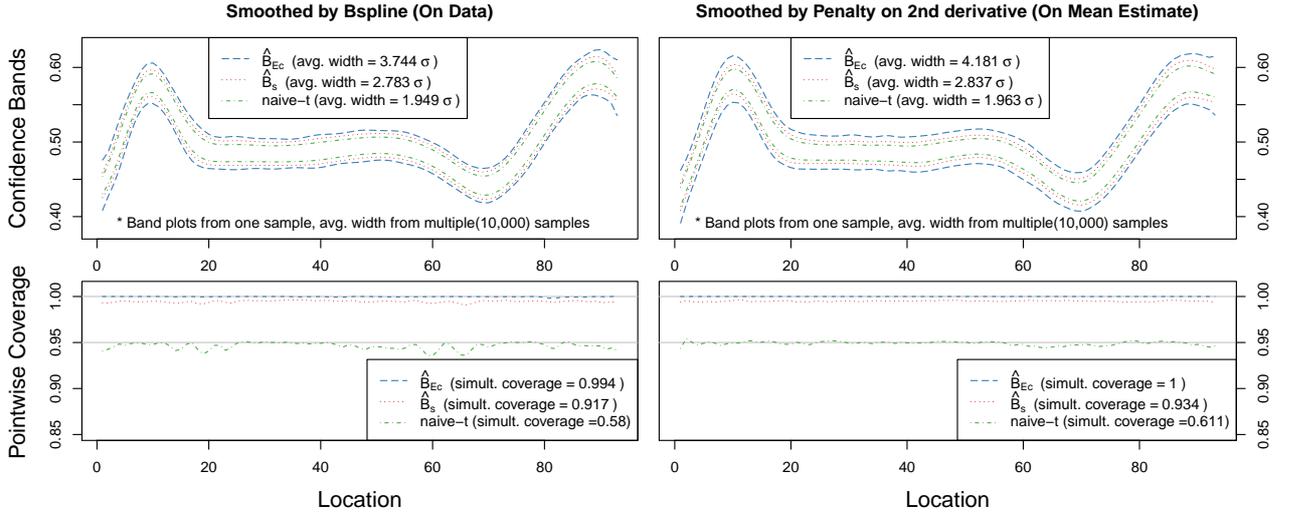}}
    \caption{\label{figure:PointwiseCoverage}The 95\% confidence bands from a simulation sample (1st row) along with point-wise coverage rates from multiple samples (2nd row).}
\end{figure}

Table \ref{tbl:ActualAlpha_Sim2} compares coverage rates of non-band form regions using two different smoothing schemes and cutting  points for $J$. 
Note that we see only minor difference between the two smoothing approaches, and the effect of $J$ is essentially the same as in Section \ref{subsection:TypeIError}; the coverage of FPC based method $\hat{E}_{PC}$ deteriorate fast as $J$ increases while $\hat{E}_{norm}$ it not affected, and $J$ that explains about $99\%$ of variance does not raise major concern in the proposed regions $\hat{E}_{c}$, $\hat{R}_{z}$, and $\hat{R}_{zs}$.

\begin{table}[h]
	\centering
	\caption{Coverage rates of non-band form regions using two different smoothing approaches}
	\label{tbl:ActualAlpha_Sim2}
	
	\begin{tabular}{r|rr|rrr|r||rr|rrr|r}
		\hline
		Smoothing & \multicolumn{6}{c||}{Bspline on the sample (15 functions)} & 
		\multicolumn{6}{c}{Penalty on the $2^{nd}$ derivative} \\
		\hline 
		var. $\geq$  & $\hat{E}_{norm}$ & $\hat{E}_{PC}$  &
		$\hat{E}_{c}$   & $\hat{R}_{z}$ & $\hat{R}_{zs}$ &  $PC^*$   
		& $\hat{E}_{norm}$ & $\hat{E}_{PC}$  &
		$\hat{E}_{c}$   & $\hat{R}_{z}$ & $\hat{R}_{zs}$ &  $PC^*$ 
		\\ \hline
		0.90  &  .949 & .942 & .948 & .947 & .952 &  5   &  .949 & .942 & .946 & .948 & .954 &  5   \\
		0.95  &  .949 & .937 & .946 & .945 & .951 &  7   &  .949 & .933 & .944 & .946 & .953 &  8   \\
		0.99  &  .949 & .904 & .940 & .939 & .947 &  11  &  .949 & .887 & .937 & .940 & .949 &  15  \\
		0.999 &  .949 & .849 & .935 & .933 & .943 &  15  &  .948 & .572 & .922 & .904 & .926 &  24  \\ \hline
	\end{tabular}
	\begin{center}
		\textit{* Median number of PCs required to capture desired (estimated) variance.}
	\end{center} 
	
\end{table}

\section{Data Example}\label{s:dti}

In this section, we further illustrate the usage of suggested methods using the same \texttt{DTI} dataset.
We now take both control and case group of the first visit scans to look at their differences in mean, in which the sample sizes are 42 and 99, respectively.

\subsection{Visualization via Bands}
\label{subsection:data_visualization}
The first step is to visually compare the two sample mean functions, and make confidence bands for the mean difference. Figure \ref{figure:DTIBand} shows the two sample means, followed by 95\% confidence bands for the mean difference using $\hat{B}_{E_c}$ and  $\hat{B}_{s}$, assuming unequal variances. 
Although the proposed band $\hat{B}_{E_c}$ is wider than $\hat{B}_{s}$ when the standard estimates from the raw data are used (middle), it gets narrower when the data are smoothed (right). For smoothing, we used quadric bsplines with two-fold cross-validations on the mean difference to choose the number of basis functions. In this case 11 basis functions were chosen and we used equally spaced knots.
We observe that the bands do not cover zero($0$) for most of the domain except for the beginning and the very end part.

\begin{figure}[h]
	\centering
	\makebox{\includegraphics[width=\textwidth]{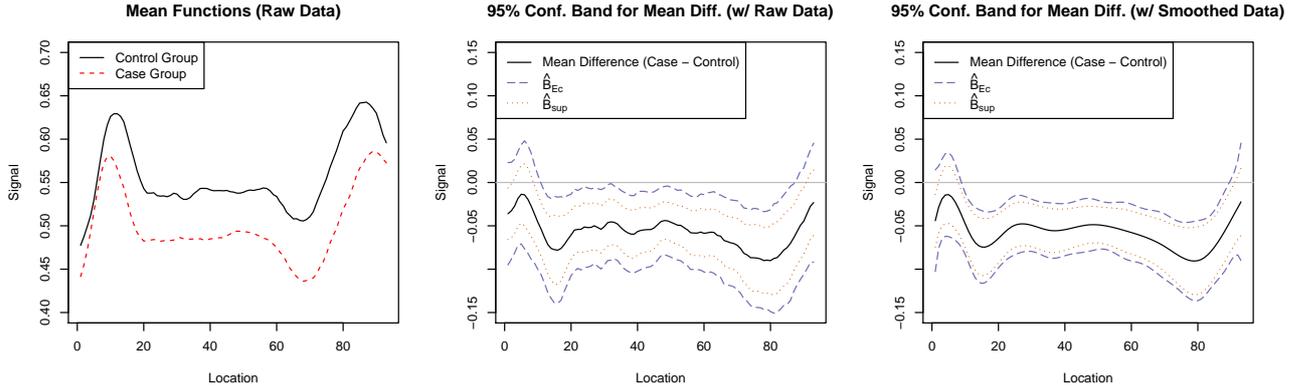}}
    \caption{\label{figure:DTIBand}The sample mean functions for the two groups (left) and the difference of these two along with $95\%$ confidence bands (middle, right). For smoothed estimates (right), the gap between $\hat{B}_{E_c}$ and $\hat{B}_{s}$ narrows down.}
\end{figure}

\subsection{Hypothesis Testing}
The result of hypothesis testing $H_0 : \mu_{\text{ctrl}} = \mu_{\text{case}}$ versus $H_0 : \mu_{\text{ctrl}} \neq \mu_{\text{case}}$ using different regions is summarized in Table \ref{tbl:DTITesting}. The proposed regions $\hat{E}_{c}$ and $\hat{R}_{z}$ yield at least  comparable p-values with existing ones like $\hat{E}_{norm}$ and $\hat{E}_{PC(3)}$. Since there exists an overall shift in the difference of the mean functions, little room could be found for the proposed regions to outperform $\hat{E}_{norm}$. Small sample version $\hat R_{zs}$ achieves a bit larger p-value as expected, but not materially. 

\begin{table}[h]
	\centering
	\caption{P-values from hypothesis testings based on different regions.}
	\label{tbl:DTITesting}
	\begin{tabular}{r|r|rrr|rrr|r}
		\hline 
		Data & Var. $\geq$  & $\hat{E}_{norm}$  & $\hat{E}_{PC}$ & $\hat{E}_{PC(3)}$ & $\hat{E}_{c}$ & $\hat{R}_{z}$ & $\hat{R}_{zs}$ & $PC^*$     \\ \hline
		Raw  & 0.99  & $6.6E^{-14}$ & $2.6E^{-10}$ & $2.1E^{-13}$ & $2.3E^{-14}$  & $2.5E^{-13}$ & $2.1E^{-11}$ & 22  \\ 
		Smoothed  & 0.99  & $1.6E^{-13}$ & $4.0E^{-14}$ & $8.7E^{-14}$ & $1.1E^{-14}$  & $2.2E^{-13}$ & $1.9E^{-11}$ & 11  \\ 
		\hline		
	\end{tabular}
	\begin{center}
		\textit{* Number of PCs used to capture desired variance except for $\hat{E}_{PC(3)}$ which uses only $3$ PCs}
	\end{center}	
\end{table}

In \citet{pomann:2016} two sample tests were developed and illustrated using the same data.  There they use a bootstrap approach to calculate p--values. A p--value of approximately zero is reported based on 5000 repetitions, which means that the p-value $< 2 \times 10^{-4}$.  Since our approach is based on asymptotic distributions, not simulations, we are able to give more precise p--values which are of the order $10^{-14}$ for the lowest and $10^{-11}$ for the highest.

\subsection{Visual Decomposition using Rectangular Region}
\label{s:rectangle-visual}
One merit of a rectangular region is that it can be expressed as intersection of marginal intervals. Note that since eigenfunctions are uniquely determined up to signs, it does not help to look at the signs of coefficients.
Figure \ref{figure:DTI_Marginals} shows confidence intervals for the absolute values of   coefficients for each PC using $\hat{R}_{z}$. We observe that only the confidence interval for the first PC does not cover zero. Based on this, we can infer that there exists a significant difference between the two mean functions along the $1^{st}$ PC, but the two means are not significantly different in any other features. In this sense, this visual decomposition serves as hypothesis testings on PCs while maintaining family-wise level at $\alpha$.
Although we made intervals for absolute coefficients to visually represent the importance of each PCs, one may choose to make intervals for absolute $z$-scores to make later intervals more visible.

\begin{figure}[h]
	\centering
	\makebox{\includegraphics[width=6.5in]{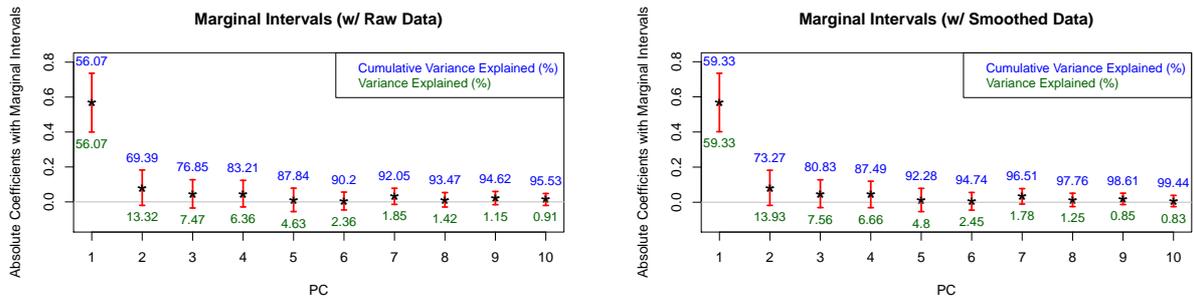}}
	\caption{\label{figure:DTI_Marginals}Confidence intervals for absolute coefficients along principal components. Each interval is centered at the absolute value of coefficient of each estimated PC, i.e. $| \langle \hat{\theta}, \hat v_j \rangle |$ for $j$-th PC, and the interval presents `reasonable candidates' for $| \langle \theta, v_j \rangle |$. The length of each interval can be used to roughly measure the importance of corresponding PC. The intervals are shown up to $10^{th}$ PC to maintain visibility but actual rectangular region using raw data used up to $22^{nd}$ PCs to capture at least $99\%$ of variance.} 
\end{figure}

Once the overall shapes of confidence intervals are obtained, one may choose to examine specific PCs.  
Figure \ref{figure:DTICoefPC1} shows the interval for the $1^{st}$ PC as a band along the $1^{st}$ eigenfunction. This now reveals that the departure is caused by the `downward' shift of the case group, and confirms that this is the main source of the mean departure in Figure \ref{figure:DTIBand}.  
Lastly, we mention that smoothing here also makes little difference in the `shapes' of the intervals in Figure \ref{figure:DTI_Marginals} and \ref{figure:DTICoefPC1} except for the effect of smoothing itself -- smoother ($1^{st}$) eigenfunction  and more variance captured in early PCs.

\begin{figure}[h]
	\centering   
	\makebox{\includegraphics[width=6.5in]{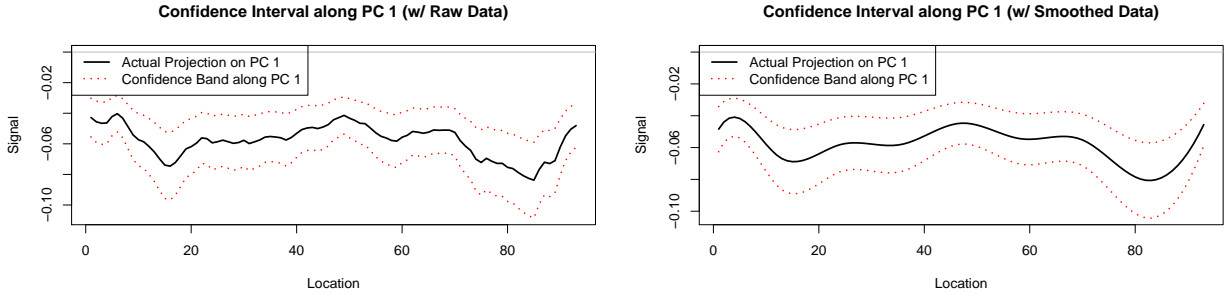}}
	\caption{\label{figure:DTICoefPC1} Representation of confidence interval for the $1^{st}$ PC as confidence band along $1^{st}$ eigenfunction.}
\end{figure}

\section{Discussion}

Each of the proposed and existing regions (and the corresponding hypothesis tests) has pros and cons, and therefore the decision on which region to use in practice would depend on many factors including the nature of the data, the purpose of the research, etc. However, we believe that we have clearly demonstrated that the proposed hyper-ellipses, $\hat E_c$, or hyper-rectangles, $\hat R_z$, make excellent candidates as the ``default" of choice. In our simulations, they were at the top or near the top, in terms of power, in every setting.  Deciding between ellipses versus rectangles comes down to how the regions will be used.  If the focus is on the principal components and interpreting their shapes, then the rectangular regions make an excellent choice.  If the FPCs are of little to no interest, then the hyper--ellipses combined with their corresponding band make an excellent choice, especially if the parameter estimate is relatively smooth.  However, for rougher estimates, we recommend sticking with the simulation based bands like \citet{degras:2011} as opposed to the bands generated from the ellipses.

We also believe that the discussed perspectives on coverage and ghosting will be useful for developing and evaluating new methodologies.  From a theoretical point of view, working with infinite dimensional parameters presents difficulties which are not found in scalar or multivariate settings.  In particular, it is common for methods to ``clip" the infinite dimensional parameters.  In practice the clipping may or may not have much of an impact -- for example the FPCA methods are very sensitive to this clipping while our ellipses and rectangles are not -- but in all cases it introduces an interesting theoretical challenge.  Our ghosting framework will be useful as it provides a sound basis for using regions with deficient coverage. 

For the first time, the construction of confidence regions and bands has been placed into a Hilbert space based framework together which has become a primary model for many FDA methodologies.   However, we believe there is a great deal of additional work to be done in this area and that it presents some exciting opportunities.  For example, are there other metrics for determining which confidence region to use?  Do these metrics lead to different choices of $c_j$?  
How can we choose $J$, the number of PCs to use in practice without undermining proper coverage considering poor estimation of later PCs?  
 Are there other shapes beyond ellipses and rectangles which are useful?  Can we use better metrics than Hausdorff for evaluating convergence?  Many open questions remain which we hope other researchers will find interesting.   


\bibliographystyle{rss}
\bibliography{jrss-b,biom1960,nb,fda}

\input bibnames.sty\hyphenation{ }\ifx \undefined \booktitle \def
  \booktitle#1{{{\em #1}}} \fi\ifx \undefined \operatorname \def \operatorname
  #1{{\rm #1}} \fi\hyphenation{Ur-i-be}\ifx \undefined \bioname \def
  \bioname#1{{{\em #1\/}}} \fi\ifx \undefined \booktitle \def
  \booktitle#1{{{\em #1}}} \fi\ifx \undefined \circled \def \circled
  #1{(#1)}\fi\ifx \undefined \insttitle \def \insttitle#1{{{\em #1}}} \fi\ifx
  \undefined \mathrm \def \mathrm #1{{\rm #1}}\fi\ifx \undefined \max \def \max
  {{\rm max}}\fi\ifx \undefined \reg \def \reg {\circled{R}}\fi
\begin{thebibliography}{24}
\expandafter\ifx\csname natexlab\endcsname\relax\def\natexlab#1{#1}\fi
\expandafter\ifx\csname url\endcsname\relax
  \def\url#1{\texttt{#1}}\fi
\expandafter\ifx\csname urlprefix\endcsname\relax\def\urlprefix{URL: }\fi

\bibitem[{Billingsley(1995)}]{billingsley:1995}
Billingsley, P. (1995) \textit{Probability and {M}easure}.
\newblock New York: Wiley, 3rd edn.

\bibitem[{Cai and Yuan(2011)}]{cai:2011}
Cai, T.~T. and Yuan, M. (2011) Optimal estimation of the mean function based on
  discretely sampled functional data: Phase transition.
\newblock \textit{The Annals of Statistics}, 2330--2355.

\bibitem[{Cao(2014)}]{cao:2014}
Cao, G. (2014) Simultaneous confidence bands for derivatives of dependent
  functional data.
\newblock \textit{Electronic Journal of Statistics}, \textbf{8}, 2639--2663.

\bibitem[{Cao et~al.(2012)Cao, Yang and Todem}]{cao:2012}
Cao, G., Yang, L. and Todem, D. (2012) Simultaneous inference for the mean
  function based on dense functional data.
\newblock \textit{Journal of Nonparametric Statistics}, \textbf{24}, 359--377.

\bibitem[{Cardot et~al.(2007)Cardot, Mas and Sarda}]{cardot:2007}
Cardot, H., Mas, A. and Sarda, P. (2007) Clt in functional linear regression
  models.
\newblock \textit{Probability Theory and Related Fields}, \textbf{138},
  325--361.

\bibitem[{Degras(2011)}]{degras:2011}
Degras, D. (2011) Simultaneous confidence bands for nonparametric regression
  with functional data.
\newblock \textit{Statistica Sinica}, \textbf{21}, 1735--1765.

\bibitem[{Feiveson and Delaney(1968)}]{feiveson:delaney:1968}
Feiveson, A.~H. and Delaney, F.~C. (1968) The distribution and properties of a
  weighted sum of chi squares.
\newblock \textit{NASA Technical Note}, \textbf{D 4575}.

\bibitem[{Goldsmith et~al.(2012a)Goldsmith, Bobb, Crainiceanu, Caffo and
  Reich}]{goldsmith:2012a}
Goldsmith, J., Bobb, J., Crainiceanu, C., Caffo, B. and Reich, D. (2012a)
  Penalized functional regression.
\newblock \textit{Journal of Computational and Graphical Statistics},
  \textbf{20}.

\bibitem[{Goldsmith et~al.(2012b)Goldsmith, Crainiceanu, Caffo and
  Reich}]{goldsmith:2012b}
Goldsmith, J., Crainiceanu, C., Caffo, B. and Reich, D. (2012b) Longitudinal
  penalized functional regression for cognitive outcomes on neuronal tract
  measurements.
\newblock \textit{Journal of the Royal Statistical Society: Series C},
  \textbf{61}, 453--469.

\bibitem[{Goldsmith et~al.(2013)Goldsmith, Greven and
  Crainiceanu}]{goldsmith:2013}
Goldsmith, J., Greven, S. and Crainiceanu, C. (2013) Corrected confidence bands
  for functional data using principal components.
\newblock \textit{Biometrics}, \textbf{69}, 41--51.

\bibitem[{Hart and Wehrly(1986)}]{Hart:1986:KRE}
Hart, J.~D. and Wehrly, T.~E. (1986) Kernel regression estimation using
  repeated measurements data.
\newblock \textit{Journal of the American Statistical Association},
  \textbf{81}, 1080--1088.

\bibitem[{Horv{\'a}th and Kokoszka(2012)}]{hkbook}
Horv{\'a}th, L. and Kokoszka, P. (2012) \textit{Inference for {F}unctional
  {D}ata with {A}pplications}.
\newblock Springer.

\bibitem[{Imhof(1961)}]{Imhof:1961:CDQ}
Imhof, J.~P. (1961) Computing the distribution of quadratic forms in normal
  variables.
\newblock \textit{Biometrika}, \textbf{48}, 419--426.

\bibitem[{Kokoszka and Reimherr(2013)}]{kokoszka:reimherr:2013b}
Kokoszka, P. and Reimherr, M. (2013) Asymptotic normality of the principal
  components of functional time series.
\newblock \textit{Stochastic Processes and their Applications}, \textbf{123},
  1546--1562.

\bibitem[{Laha and Roghatgi(1979)}]{laha:roghatgi:1979}
Laha, R.~G. and Roghatgi, V.~K. (1979) \textit{Probability {T}heory}.
\newblock Wiley.

\bibitem[{Li and Hsing(2010)}]{li:hsing:2010}
Li, Y. and Hsing, T. (2010) Deciding the dimension of effective dimension
  reduction space for functional and high-dimensional data.
\newblock \textit{The Annals of Statistics}, \textbf{38}, 3028--3062.

\bibitem[{Pomann et~al.(2016)Pomann, Staicu and Ghosh}]{pomann:2016}
Pomann, G.-M., Staicu, A.-M. and Ghosh, S. (2016) A two-sample
  distribution-free test for functional data with application to a diffusion
  tensor imaging study of multiple sclerosis.
\newblock \textit{Journal of the Royal Statistical Society: Series C (Applied
  Statistics)}.

\bibitem[{Ramsay and Silverman(2005)}]{ramsay:silverman:2005}
Ramsay, J.~O. and Silverman, B.~W. (2005) \textit{Functional {D}ata
  {A}nalysis}.
\newblock Springer.

\bibitem[{Reimherr(2015)}]{reimherr:2015}
Reimherr, M. (2015) Functional regression with repeated eigenvalues.
\newblock \textit{Statistics \& Probability Letters}, \textbf{107}, 62--70.

\bibitem[{Reimherr and Nicolae(2014)}]{ReNi:2014}
Reimherr, M. and Nicolae, D. (2014) A functional data analysis approach for
  genetic association studies.
\newblock \textit{The Annals of Applied Statistics}, \textbf{9}, 406--429.

\bibitem[{Rudin(1976)}]{rudin:1976}
Rudin, W. (1976) \textit{Principles of {M}athematical {A}nalysis}.
\newblock Singapore: McGraw-Hill, third edn.

\bibitem[{Yao et~al.(2005)Yao, M{\" u}ller and Wang}]{yao:muller:wang:2005JASA}
Yao, F., M{\" u}ller, H.-G. and Wang, J.-L. (2005) Functional data analysis for
  sparse longitudinal data.
\newblock \textit{Journal of the American Statistical Association},
  \textbf{100}, 577--590.

\bibitem[{Zhang and Wang(2016)}]{zhang:wang:2016}
Zhang, X. and Wang, J. (2016) From sparse to dense functional data and beyond.
\newblock \textit{The Annals of Statistics}, Forthcoming.

\bibitem[{Zheng et~al.(2014)Zheng, Yang and H{\"a}rdle}]{zheng:2014}
Zheng, S., Yang, L. and H{\"a}rdle, W. (2014) A smooth simultaneous confidence
  corridor for the mean of sparse functional data.
\newblock \textit{Journal of the American Statistical Association},
  \textbf{109}, 661--673.

\end{thebibliography}

\newpage

\appendix

\section{Proofs}	
In this section we gather all of the proofs and necessary lemmas.

\begin{proof}[Proof of Theorem \ref{t:band}]
	The $\mcH$ norm of $r(x)$ is given by 
	\[
	\| r\|^2 = \frac{\xi}{N}\sum_{j=1}^\infty c_j^2.
	\]
	This will be finite if $\sum_{j=1}^\infty c_j^2 < \infty$ and $|\xi| < \infty$, the latter of which is guaranteed when  $\sum \lambda_j c_j^{-2}<\infty $.  Therefore $r(x)$ is in $\mcH$. 
	
	To show $E_{\hat{\theta}} \subset B_{\hat{\theta}}$, take $h \in E_{\hat{\theta}}$. Using the Cauchy-Schwartz inequality and \eqref{e:EllipsoidBound}, we get  
	\begin{align*}
	\left( h(x) - \hat{\theta}(x) \right)^2 &= \left(\sum_{j=1}^\infty \langle h-\hat{\theta}, v_j \rangle v_j(x) \right)^2 
	= \left(\sum_{j=1}^\infty \frac{\langle \sqrt{N}(h-\hat{\theta}), v_j \rangle}{c_j} \frac{1}{\sqrt{N}}c_jv_j(x)  \right)^2 \\
	& \leq  \sum_{j=1}^\infty \frac{\langle \sqrt{N}(h-\hat{\theta}), v_j \rangle^2}{c_j^2}   \sum_{j=1}^{\infty} \frac{1}{N}c_j^2 v_j^2(x) \\  
	& \leq \sum_{j=1}^{\infty} \frac{\xi}{N}c_j^2 v_j^2(x)  \equiv r^2(x), 
	\end{align*}
	for $x$ \textit{almost everywhere}, which then implies $h \in B_{\hat{\theta}}$ and thus  $E_{\hat{\theta}} \subset B_{\hat{\theta}}$ as desired.  

\end{proof}

\begin{lemma}
	\label{lemma:DiffinNormExpansion}
	Define $\alpha_1 := \lambda_1 - \lambda_2$ and $\alpha_j := \min\{ \lambda_j - \lambda_{j+1}, \lambda_{j-1} - \lambda_{j} \}$ for $j = 2, \dots$.  Then with probability 1
	\begin{align}
	\label{e:DiffinNormExpansion}
	\| \hat v_j - v_j \| \leq \frac{ 2 \sqrt{2} \| \hat C_\theta - C_\theta\|_{\mcL}}{ \alpha_j}
	\qquad \text{and} \qquad
	| \hat \lambda_j - \lambda_j| \leq \| \hat C_\theta - C_\theta\|_{\mcL}.
	\end{align}	
	\begin{proof}
		See \citet{hkbook}.
	\end{proof}
\end{lemma}

\begin{proof}[Proof of Theorem \ref{thm:Convergence1}]
	Our aim is to show that for any $x \in \hat{E}_{\hat{\theta}}$ there exists $y \in E_{\hat{\theta}}$, s.t. $\| y - x\|$ is bounded by the RHS of \eqref{e:order1}. 
	For any $x \in \hat{E}_{\hat{\theta}}$, take $y \in \cH$ s.t., $y = \hat{\theta}  + \sum_{j=1}^{J} \langle  x - \hat{\theta}, \hat{v}_j \rangle v_j$.
	We then have that
	$$\sum_{j=1}^\infty \frac{\langle y - \hat{\theta} , v_j\rangle^2}{N^{-1} c_j^2} 
	= \sum_{j=1}^{J} \frac{\langle x - \hat{\theta} , \hat{v}_j\rangle^2}{N^{-1}c_j^2}
	\leq \xi,$$
	which implies that $y \in E_{\hat{\theta}}$ follows from $x \in  \hat{E}_{\hat{\theta}}.$  Turning to the difference between $x$ and $y$ we have that
	\begin{align*}
	\| y - x \|^2
	&= \| (y - \hat{\theta}) - (x - \hat{\theta}) \|^2 \\
	&= \left\| \sum_{j=1}^{J} \langle  x - \hat{\theta}, \hat{v}_j \rangle v_j - \sum_{j=1}^{J} \langle  x - \hat{\theta} , \hat{v}_j \rangle \hat{v}_j \right\|^2 
	= \left\| \sum_{j=1}^{J} \langle  x- \hat{\theta} , \hat{v}_j \rangle (v_j - \hat{v}_j) \right\|^2.
	\end{align*} 
	From 
	Cauchy-Schwarz inequality, the above is bounded by
	\[
	\sum_{j=1}^J \langle  x- \hat{\theta} , \hat{v}_j \rangle ^2 \sum_{j=1}^J \| v_j - \hat v_j\|^2
	\leq N^{-1} \xi  c_1^2 \sum_{j=1}^J \| v_j - \hat v_j\|^2,
	\]
	which holds uniformly in $x$.  
	Using Lemma \ref{lemma:DiffinNormExpansion} we get
	\begin{align*}
	\sum_{j=1}^{J} \left\| v_j - \hat{v}_j \right\|^2 
	& \leq  \sum_{j=1}^{J}  \frac{8 \|\hat C_\theta - C_\theta\|_\mcL^2 }{\alpha_j^{2}}. 
	\end{align*}
	Therefore, 
	\begin{align*}
	\rho(\hat{E}_{\hat{\theta}},E_{\hat{\theta}}) ^2 \leq  \sum_{j=1}^{J}  \frac{8 \xi  c_1^2 \|\hat C_\theta - C_\theta\|_\mcL^2 }{N \alpha_j^{2}}
	\end{align*}
	as claimed.
\end{proof}

\begin{proof}[Proof of Theorem \ref{thm:Convergence2}]
	Again, we aim to show that for any $y \in E_{\hat{\theta}}$ there exists $x \in \hat{E}_{\hat{\theta}}$, such that $\| y - x\|$ achieves the claimed bound. Take $y_J := \hat{\theta} + \sum_{j=1}^J \langle y - \hat{\theta}, v_j \rangle v_j$, and $x := \hat{\theta} + \sum_{j=1}^J \langle y - \hat{\theta}, v_j \rangle \hat{v}_j$.  As before, $x \in \hat{E}_{\hat{\theta}}$ follows from $y \in E_{\hat{\theta}}$.  
	We then use a triangle inequality to obtain
	\begin{align}
	\label{e:local1}
	\|y - x\| \leq \|y - y_J \| + \|y_J - x \|. 
	\end{align}
	The first term is bounded by
	\begin{align*}
	\|y - y_J\|^2 
	& = \sum_{j=J+1}^\infty \langle y - \hat{\theta}, v_j\rangle^2 
	= \sum_{j=J+1}^\infty c_j^2 \frac{\langle y - \hat{\theta}, v_j\rangle^2}{c_j^2} 
	\leq c_{J}^2  \sum_{j=J+1}^\infty  \frac{\langle y - \hat{\theta}, v_j\rangle^2}{c_j^2} 
	\leq  c_J^2  N^{-1} \xi,
	\end{align*} 
	uniformly in $y$.  Using the same arguments as in Theorem \ref{thm:Convergence1} we obtain
	\begin{align*}
	\|y_J - x\|^2  & = \left\| \sum_{j=1}^J \langle y - \hat{\theta} , v_j \rangle (v_j - \hat{v}_j) \right\|^2 
	\leq   \sum_{j=1}^{J}  \frac{8 \xi  c_1^2 \|\hat C_\theta - C_\theta\|_\mcL^2 }{N \alpha_j^{2}} 
	\end{align*}
	uniformly in $y$ as well.  Therefore, 
	\begin{align*}
	\rho(E_{\hat{\theta}}, \hat{E}_{\hat{\theta}})  = & \left[c_J^2  N^{-1} \xi \right]^{\frac{1}{2}}+ 
	\left[ \sum_{j=1}^{J}  \frac{8 \xi  c_1^2 \|\hat C_\theta - C_\theta\|_\mcL^2}{N \alpha_j^{2}} \right]^{\frac{1}{2}}
	\end{align*}
	as claimed.
\end{proof}

\begin{lemma} \label{l:sum}
	Let $\alpha > 0$, then as $J \to \infty$
	\[
	\sum_{j = 1}^J j^{\alpha} \approx \frac{J^{\alpha + 1}}{\alpha+1}
	\quad \text{ i.e. } \quad 
	\frac{\sum_{j = 1}^J j^{\alpha}}{J^{\alpha + 1}/(\alpha+1)  } \to 1.
	\]
\end{lemma}
\begin{proof}
	We can rewrite
	\[
	\sum_{j = 1}^J j^{\alpha} = J^{\alpha + 1} \sum_{j = 1}^J \frac{1}{J } \left(\frac{j}{J}\right)^{\alpha}. 
	\]
	Using the definition of the Riemann integral we have that
	\[
	J^{\alpha + 1} \sum_{j = 1}^J \frac{1}{J } \left(\frac{j}{J}\right)^{\alpha} \approx J^{\alpha + 1} \int_{0}^1 x^{\alpha} \ dx = 
	\frac{J^{\alpha + 1}}{\alpha + 1},
	\]
	which is the desired result.
\end{proof}
\begin{proof}[Proof of Theorem \ref{thm:ConvergenceOverall}]
	We begin by analyzing the sum of the $\alpha_j^{-2}$.  Applying Assumption \ref{a:rate1} and Lemma \ref{l:sum} we have that
	\begin{align*}
	\sum_{j=1}^J \alpha_j^{-2}  & \leq K  \sum_{j=1}^J  j^{2\delta+2} 
	\approx \frac{K}{2\delta+3} J^{2\delta+3}.
	\end{align*}
	This implies that \eqref{e:order1} is bounded by 
	\begin{align}
	\label{e:rate1.1}
	\mbE \rho(\hat E_{\hat \theta}, E_{\hat \theta})^2 
	\leq \sum_{j=1}^J \frac{8 \xi c_1^2 \mbE \| \hat C_\theta - C_\theta \|_{\mcL}^2}{N \alpha_j^2}
	\leq   J^{{2 \delta + 3}} N^{-2} O(1).
	\end{align} 
	The second distance can be bounded using the simple scalar relationship $(a + b)^2 \leq 2a^2 + 2 b^2$, which gives
	\[
	\rho(E_{\hat \theta}, \hat E_{\hat \theta})^2
	\leq 2 c_J^2 N^{-1} \xi +  \sum_{j=1}^J \frac{16 \xi c_1^2  \| \hat C_\theta - C_\theta\|_{\mcL}^2}{N \alpha_j^2}.
	\]
	Using Assumption \ref{a:rate1} we then have that
	\[
	\mbE[\rho(E_{\hat \theta}, \hat E_{\hat \theta})^2]
	= J^{-2\gamma} N^{-1} O(1)  + J^{{2 \delta + 3}} N^{-2} O(1).
	\]
	Setting the two errors equal to each other yields
	\[
	J^{-2\gamma} N^{-1} = J^{{2 \delta + 3}} N^{-2}
	\Longrightarrow J  = N^{\frac{1}{2\delta  + 3 + 2 \gamma}}.
	\]
	This yields an overall error of 
	\[
	\mbE d_H(E_{\hat \theta}, \hat E_{\hat \theta})^2
	\leq N^{- \left(2  -  \frac{2\delta  + 3 }{2\delta + 3 + 2\gamma} \right)} O(1) .
	\]
	as claimed.
\end{proof}

\begin{proof}[Proof of Corollary \ref{corol:EchXConvRate1}]
	Recall that $c_j^2 = \hat \lambda_j^{1/2}$.  Denote $\tilde c_j^2 =\lambda_j^{1/2}$ and the resulting $J$ dimensional confidence region as $\tilde E_{\hat \theta}$.  Here $\tilde E_{\hat \theta}$ acts an intermediate step between $\hat E_{\hat \theta}$ and $E_{\hat \theta}$. 
	When using $\tilde E_{\hat \theta}$ we have that $ 2 \gamma = \delta /2 $.  The rate in Theorem \ref{thm:ConvergenceOverall} then has an exponent of
	$$
	- \left( 2 - \frac{2\delta + 3}{2\delta + 3 + \delta/2}\right)
	= - \frac{6\delta + 6}{5 \delta + 6},
	$$
	which means that
	\[
	d_H(\tilde E_{\hat \theta}, E_{\hat \theta})^2 = O_P( N^{ - \frac{6\delta + 6}{5 \delta + 6}}).
	\]
	We will now show that the distance $ d_H(\hat E_{\hat \theta}, \tilde E_{\hat \theta})^2$ is of a smaller order, which implies that $d_H(\hat E_{\hat \theta}, E_{\hat \theta})^2$ has the same rate as $ d_H(\tilde E_{\hat \theta},  E_{\hat \theta})^2$, as desired.
	
	Recall that if $x \in \hat E_{\hat \theta}$ then it satifies
	\[
	\sum_{j=1}^J \frac{\langle x - \hat \theta, \hat v_j \rangle^2 }{N^{-1} c_j^2 } < \xi.
	\]
	We can create a scaled $x$, call it $\tilde x$ such that it is also in $\tilde E_{\hat \theta}$ by noticing that
	\[
	\sum_{j=1}^J \frac{\langle x - \hat \theta, \hat v_j \rangle^2 }{N^{-1} \tilde c_j^2 }
	= \sum_{j=1}^J \frac{\langle x - \hat \theta, \hat v_j \rangle^2 }{N^{-1} c_j^2 } \frac{c_j^2}{\tilde c_j^2}
	\leq \max_{j=1,\dots,  J} \frac{c_j^2}{\tilde c_j^2} \sum_{j=1}^J \frac{\langle x - \hat \theta, \hat v_j \rangle^2 }{N^{-1} c_j^2 } 
	\leq  \xi \max_{j=1,\dots,  J} \frac{c_j^2}{\tilde c_j^2}.
	\]
	So $\tilde x \in \tilde E_{\hat \theta}$ if $ \tilde x = \hat \theta +  \left(\max_{j=1,\dots,  J} c_j^2 / \tilde c_j^2 \right)^{-\frac{1}{2}} (x - \hat \theta)$.  Now the difference between $x$ and $\tilde x$ is given by, using the simple scalar relationship $(1-\sqrt{a})^2 \leq (1-a)^2$ for $a > 0$, 
	\begin{align*}
	\|x - \tilde x\|^2 = \sum_{j=1}^J \langle x - \tilde x, \hat v_j \rangle^2 & =  \left(1 - \left(\max_{j=1,\dots,  J} c_j^2 / \tilde c_j^2 \right)^{-\frac{1}{2}}\right)^2
	\sum_{j=1}^J \langle x - \hat \theta, \hat v_j \rangle^2  \\
	& \leq  \left(1 - \left(\max_{j=1,\dots,  J} c_j^2 / \tilde c_j^2 \right)^{-1}\right)^2
	\sum_{j=1}^J \langle x - \hat \theta, \hat v_j \rangle^2  \\
	& \leq \xi N^{-1} c_1^2  \left(1 - \left(\max_{j=1,\dots,  J} c_j^2 / \tilde c_j^2 \right)^{-1}\right)^2 \\
	& \leq \xi N^{-1} c_1^2
	\max_{j=1,\dots,  J} \frac{(c_j^2 - \tilde c_j^2)^2}{c_j^4} \\
	& \leq \frac{\xi N^{-1} c_1^2}{c_J^{4}}
	\max_{j=1,\dots,  J} (c_j^2 - \tilde c_j^2)^2. 
	\end{align*}
	Using a Taylor expansion and Lemma \ref{lemma:DiffinNormExpansion} one has that $ \max_{j=1,\dots,  J} (c_j^2 - \tilde c_j^2)^2 = O_P(N^{-1})$ and $c_J^{-4} = O_P(\lambda_J^{-1})$.   We therefore have that
	\[
	\rho(\hat E_{\hat \theta}, \tilde E_{\hat \theta})^2
	= \frac{1}{N^2 \lambda_J}O_P(1).
	\]
	Plugging in the optimal $J$ we get that
	\[
	\rho(\hat E_{\hat \theta}, \tilde E_{\hat \theta})^2 = 
	J^{\delta}N^{-2}  O_P(1)
	= N^{- \left(2 - \frac{\delta}{2 \delta + 3 + \delta/2} \right)} O_P(1)
	=  N^{- \frac{8 \delta +12 }{5 \delta + 6} } O_P(1)
	= N^{- \frac{6 \delta +6 }{5 \delta + 6} } o_P(1).
	\]
	Nearly identical arguments will yield the same result for the reverse $\rho(\tilde  E_{\hat \theta}, \hat E_{\hat \theta})^2$.  
	Thus $d(\hat  E_{\hat \theta}, \tilde E_{\hat \theta})^2$ is of a lower order than $d(\tilde E_{\hat \theta},  E_{\hat \theta})^2$ and the claim holds.
\end{proof} 

\section{Other Criteria}
\label{section:OtherCriteria}

Both in hyper-ellipsoid and hyper-rectangular regions, there exists infinitely many options to find $\{c_j\}$ that determines their shapes. The following introduces a few more criteria that may be found to be interesting.

\subsection{Hyper-Ellipsoid}

In hyper-ellipsoid, one may take 
\[ c_j^2 = \left( \sum_{i=j}^\infty \lambda_i \right) ^ {1/2}, \]
i.e. the square root of the tail sum of the eigenvalues.  Recall that $\sum_{j=1}^\infty \lambda_j < \infty$ since it is equal to the trace of the covariance operator. It is therefore clear that $c_j \to 0$ and the region is compact.  What is not as obvious is that one also has  $\sum_{j=1}^\infty {\lambda_j}{c^{-2}_j} < \infty$, which means that the resulting $W_\theta$ is a  random variable with finite mean and the resulting region can obtain the proper coverage.  Showing this is actually an interesting real analysis exercise and we refer the reader to \citet[p. 80]{rudin:1976} for more details. We denote this region as $E_{c1}$

\subsection{Hyper-Rectangle}
For hyper-rectangular regions, the first possibility is to use the same $\{c_j\}$ used in hyper-ellipsoid and find $\{z_j\}$ numerically. Since $\{z_j\}$ is uniquely determined by $\xi$ once $\{c_j\}$ is given, one can easily search for $\xi$ that satisfies \eqref{e:RectangleProbRule}. We will denote the region achieved in this way using $c_j^2 = \left( \sum_{i=j}^\infty \lambda_i \right) ^ {1/2}$ of $E_{c1}$ as $R_{c1}$, and the one uses $c_j^2 = \lambda_j^{1/2}$ of $R_{c}$ as $E_{c}$

Next approach we considered is to find $\{z_j\}$ that minimizes $\sup \{ \| h - \hat{\theta} \|^2 : h \in R_{\hat{\theta}} \}$, i.e. the distance between the farthest point of the region from the center. It is equivalent to finding a rectangular region that has smallest $\mcH$ norm. The problem reduces to minimizing $\frac{1}{N} \sum_{j=1}^{\infty} \lambda_j z_j^2$ under the constraint $\sum_{j=1}^{\infty}\log \Phi_{sym}(z_j) = \log(1-\alpha)$. The solution for this problem can be found in the Appendices Section \ref{Appendix:FindZ1}, and the following summarizes the steps to follow :
\begin{enumerate}
	\item Define a function $f(z) := \exp(z^2 / 2) \Phi_{sym}(z)$
	\item Define $f^{-1}(\cdot)$, the inverse of $f(\cdot)$. This can be achieved numerically by univariate optimization in practice.
	\item Find $M^* = \arg\min_{M} 
	\left| \prod_{j} \Phi_{sym}\left\{ f^{-1} \left( \frac{M}{\sqrt{2\pi} \lambda_j}\right) \right\} - (1-\alpha) \right|$. This also can be achieved numerically by univariate optimization.
	\item Take $z_j = f^{-1}\left( \frac{M^*}{\sqrt{2\pi}\lambda_j}\right)$ for each $j$.
\end{enumerate}
We will denote this region as $R_{z1}$

Yet another approach, which is much simpler, is to start from $\sum_{j=1}^{\infty}\log \Phi_{sym}(z_j) = \log(1-\alpha) $ and distribute $\log(1-\alpha)$ among $\Phi_{sym}(z_j)$'s, possibly assigning more weight to early $j$'s to narrow down length along the early eigenfunctions. For example, $\{z_j\}$ that satisfies $\log \Phi_{sym}(z_j) = \frac{\lambda_j^\rho}{\sum_{k=1}^\infty \lambda_k^\rho} \log(1-\alpha)$ is an intuitive option as long as $\sum_j \lambda_j^\rho < \infty$ is satisfied for some positive $\rho$. This then has a closed form solution $z_j = \Phi_{sym}^{-1}\left[ \exp \left( {\frac{\lambda_j^\rho}{\sum_{k=1}^\infty \lambda_k^\rho} \log(1-\alpha)} \right) \right]$ for each $j$. Therefore,  any $\rho \geq 1$ can be used regardless the smoothness of the process, and $\rho = 1$ leads to what we denoted as $R_{z}$ in the main text. We empirically observed that $\{z_j\}$ found this way is similar to $\{z_j\}$ found from $R_{z1}$ above.

\section{Small Sample Version}
\label{Appendix:SmallSample}
When $C_\theta$ is unknown, we relied on the consistency of $\hat{C}_\theta$ and therefore replaced $\{( v_j , \lambda_j) \}_{j=1}^{\infty}$ with $\{( \hat{v}_j , \hat{\lambda}_j) \}_{j=1}^{J}$ to construct empirical versions of confidence regions. Although this approach is still valid, one might want to look at an alternative if $N$ is too small. We suggest here a simple technique to respond to this concern. To utilize an explicit form of $\hat{C}_\theta$, we will turn our attention to a special case of $\hat{\theta}$ in this section only. Note, however, that same idea can be applied in other situations too. 
Consider $\hat{\theta} = \bar{X}$, where $\bar{X} = N^{-1}\sum_{i=1}^N  X_i$, and  $X_i \stackrel{iid}{\sim} \mcN(\theta, C_\theta)$. The standard covariance estimator $\hat{C}_\theta$ is
\begin{align*}
\hat{C}_\theta = \frac{1}{N-1} \sum_{i=1}^{N} (X_i - \bar{X}) \otimes  (X_i - \bar{X}).
\end{align*}
By KL expansion \eqref{e:KL}, we can write $X_i = \theta + \sum_{j=1}^\infty \sqrt{\lambda_j}Z_{ij}v_j$ where $Z_{ij} \stackrel{iid}{\sim}\mcN(0,1)$. Denoting $\bar{Z}_{\cdot j} := N^{-1}\sum_{i=1}^N Z_{ij}$,
\begin{align*}
\hat{C}_\theta &= \frac{1}{N-1} \sum_{i=1}^N \left( \sum_{j=1}^\infty \sqrt{\lambda_j}(Z_{ij} -\bar{Z}_{\cdot j}) v_j \right) \otimes
\left( \sum_{k=1}^\infty \sqrt{\lambda_{k}}(Z_{ik} -\bar{Z}_{\cdot k}) v_{k} \right) \\
&= \sum_{j=1}^\infty \sum_{k=1}^\infty \sqrt{\lambda_j \lambda_{k}} \frac{1}{N-1} \sum_{i=1}^N (Z_{ij} -\bar{Z}_{\cdot j})(Z_{ik} -\bar{Z}_{\cdot k}) \left( v_j \otimes v_{k} \right) 
\end{align*}
gives basis expansion of $\hat{C}_\theta$ using $\{(v_j \otimes v_k)\}_{j,k}$, an orthonormal basis of $\mcS$, other than the common expansion of $\hat{C}_\theta = \sum_{j=1}^{N-1} \hat{\lambda}_j (\hat{v}_j \otimes \hat{v}_j)$.

Therefore, the coefficient of $\hat{C}_\theta$ with respect to $(v_j \otimes v_j)$ becomes
\begin{align}
\label{e:tildeLambda}
\tilde{\lambda}_j := \langle \hat{C}_\theta, (v_{j} \otimes v_{j}) \rangle_{\mcS} = \frac{\lambda_j}{N-1} \sum_{i=1}^{N} (Z_{ij} - \bar{Z}_{\cdot j})^2
\end{align}
Observer that $V_j := \sum_{i=1}^{N} (Z_{ij} - \bar{Z}_{\cdot j})^2 \stackrel{d}{=} \chi^2_{N-1}$ and is independent from $Z_j := \sqrt{N }\bar{Z}_{\cdot j} \stackrel{d}{=} \mcN(0,1)$. Therefore, $T_j := Z_j / \sqrt{V_j/(N-1)}$ follows $t$ distribution with $N-1$ degree of freedom, which are also mutually independent among  $j$'s. Finally, using $\sqrt{\lambda_j} = \sqrt{\tilde{\lambda}_j}/\sqrt{V_j / (N-1)}$ in \eqref{e:tildeLambda}, we achieve expansion 
\begin{align}
\begin{split}
\label{e:ExpansionT}
\sqrt{N} (\hat{\theta} - \theta) 
= \sum_j \sqrt{\lambda_j}Z_j v_j 
= \sum_j \sqrt{\tilde{\lambda}_j} T_j v_j.
\end{split}
\end{align}

Expansion \eqref{e:ExpansionT} implies that we can work with exact distribution  with the knowledge of eigenfunctions, not both eigenfunctions and eigenvalues. In practice where eigenfunctions are not known, we will still replace  $\{(\tilde{\lambda}_j, v_j)\}$ with $\{(\hat{\lambda}_j, \hat{v}_j)\}$. Note, however, that this approach now depend only on the consistency of $\{\hat{v}_j\}$, not those of both $\{\hat{v}_j\}$ and $\{\hat{\lambda}_j\}$. Simulation study in Section \ref{section:Simulation} confirms that this approach is actually appealing.

The rest of this section discusses how to obtain confidence regions utilizing expansion $\eqref{e:ExpansionT}$. All the following implementations need replacement of $\{\lambda_j\}$ with $\{\tilde{\lambda}_j\}$ when $\{v_j\}$ is known, although $\{\hat{\lambda}_j\}$ will be used for  both $\{\lambda_j\}$ and $\{\tilde{\lambda}_j\}$ when $\{v_j\}$ or covariance operator is unknown.

\paragraph{Hyper-Rectangular Regions} For hyper-rectangular regions, one can consider (at least) two options. First option is to maintain the same ratio of radii. One replace $z_j$ with $t_j := c z_j$, where $c > 1$ can be numerically found to satisfy $\prod_{j} \mbP(|T_j| \leq t_j) = 1-\alpha$. Another option is to take $\{t_j\}$ such that $\mathbb{P} (|Z_j| \leq z_j) = \mathbb{P} (|T_j| \leq t_j)$ for each $j$. Favoring the simplicity of implementation, we used the second option in our simulation study, and was denoted with suffix `$s$', for example, as $R_{zs}$. Note, however, that the first option is not costly either.

\paragraph{Hyper-Ellipsoid Regions} For a hyper ellipsoid region, we observe that  
\begin{align*}
W_{\theta} = \sum_j \frac{\lambda_j Z_j^2}{c_j^2}
= \sum_j \frac{\tilde{\lambda}_j }{c_j^2} Z_j^2 \frac{\lambda_j}{\tilde{\lambda_j}}
= \sum_j \frac{\tilde{\lambda}_j }{c_j^2} \frac{Z_j^2}{V_j / (N-1)}. 
\end{align*}
By defining $F_j := \frac{Z_j^2}{V_j / (N-1)}$, $W_\theta = \sum_j \frac{\tilde{\lambda}_j }{c_j^2} F_j$ follows weighted sum of independent $\mathcal{F}_{1,N-1}$ distributions, with weights being $\{\tilde{\lambda}_j c_j^{-2}\}$. This region was not implemented  since no (numerical) tool for the distribution of sum of weighted $\mathcal{F}$ seemed to be available yet.

\section{Finding $\{z_j\}$ for $R_{z1}$}
\label{Appendix:FindZ1}

Under the constraint of $\sum_{j=1}^{\infty}\log \Phi_{sym}(z_j) = \log(1-\alpha)$, we want to find 
\begin{align*}
\left\{z_{j}\right\} 
&= \arg \min_{\left\{z_{j}\right\}} \sup \{ \| h - \hat{\theta} \|^2 : h \in S_{\hat{\theta}} \}  
= \arg \min_{\left\{z_{j}\right\}} \frac{1}{N} \sum_{j=1}^{\infty} \lambda_j z_j^2   
= \arg \min_{\left\{z_{j}\right\}} \sum_{j=1}^{\infty} \lambda_j z_j^2. 
\end{align*}
Take 
$
L(M, z_1, \cdots) :=  \sum_{j} \lambda_j z_j^2 - M \left( \sum_{j}  \log \Phi_{sym}(z_j) - \log (1-\alpha) \right)
$
where $M$ is Lagrange multiplier. To minimize $L$, set
\begin{align*}
\frac{\partial}{\partial{z_j}} L &= 2\lambda_j z_j - 
M \{\Phi_{sym}(z_j)\}^{-1} \frac{\partial}{\partial{z_j}}\int_{-z_j}^{z_j}(2\pi)^{-1/2}e^{-x^2/2}dx   \\
&= 2\lambda_j z_j - M \{\Phi_{sym}(z_j)\}^{-1} \sqrt{\frac{2}{\pi}}e^{-z_j^2/2} = 0.
\end{align*}
We then achieve 
$
M = \sqrt{2\pi} \lambda_j e^{z_j^2/2}z_j \Phi_{sym}(z_j)
$
for each $j$. By defining a function $f(z) := e^{z^2/2}z\Phi_{sym}(z)$ which is increasing in $z > 0$, 
we can write $z_j$ as 
$z_j = f^{-1}\left(M / \sqrt{2\pi}\lambda_j \right)$. 
Therefore, each $z_j$ is can be uniquely determined once $M$ is found. We can  numerically search for such $M$ that satisfies 
\begin{align}
-\frac{\partial}{\partial{M}} L &= \sum_{j}  \log \Phi_{sym}(z_j) - \log (1-\alpha)= 0
\end{align}
This is not computationally difficult nor expensive since $\frac{\partial}{\partial{M}} L$ is monotone in $M$ and we can easily find upper and lower bounds for $M$ from those for $z_1$. For the range of $z_1$, we require $z_1 \geq \Phi_{sym}^{-1}(1-\alpha) \equiv z_{1,lb}$ where  $\Phi_{sym}^{-1}$ is inverse of $\Phi_{sym}$. We also want $z_1$ to be smaller than at least average of $z_i$'s, so we can take $z_{1,ub}$, the upper bound of $z_1$, to satisfy $\prod_{j=1}^p \Phi_{sym}(z_{1,ub}) = (\Phi_{sym}(z_{1,ub}))^p = 1-\alpha$,  i.e., $z_{1,ub} = \Phi_{sym}^{-1}\left\{(1-\alpha)^{1/p}\right\}$. This gives a practical range for $M$ as
$
\sqrt{2\pi} \lambda_1 f(z_{1,lb})  \leq M \leq \sqrt{2\pi} \lambda_1 f(z_{1,ub}) 
$.

\paragraph{Finding p-value} Observe that $M = \sqrt{2\pi} \lambda_j f(z_j)$ hold for all $j$. We can then take $M^\star :=  \sup_j \{\sqrt{2\pi} \lambda_j f(z_j)\}$ as our statistic to get p-value from, namely $1 - \prod_{j=1}^p \Phi_{sym}(z_j^\star)$ where $z_j^\star := f^{-1}\left(M^\star / \sqrt{2\pi}\lambda_j \right)$.

\section{Calculating $P$-values for HT from Confidence Regions}
\label{supp:PvalueRectangle}

This section describes steps to find $p$-values from the hypothesis testing baseds on proposed hyper-ellipsoid and hyper-rectangular regions. Note that $p$-value can be interpreted as the smallest $\alpha$ that makes the confidence region to cover $\theta_0$. For estimated regions, we replace $\{(v_j, \lambda_j)\}$ or $\{(v_j, \tilde{\lambda}_j)\}$ with $\{(\hat{v}_j, \hat{\lambda}_j)\}$.

\subsection{Hyper-ellipsoid Regions}
We find the observed test statistic $W^* := \sum_{j}  N \langle \hat{\theta} - \theta_0 , v_j \rangle^2 / c_j^2$ and get $p$-value as $\mbP \left(W_\theta \geq W^* \right)$ where $W_\theta$ is a weighted sum of $\chi^2$ random variables with weights $\{\lambda_j/c_j^2\}$.

\subsection{Hyper-rectangular Regions}

For hyper-rectangular regions, we first need to find z-score for each $j$ as $z_j^* = \langle \sqrt{N}(\hat{\theta} - \theta_0), v_j \rangle / \sqrt{\lambda_j}$. The next step differs by each criterion.

\begin{enumerate}
	\item When $\{c_j\}$ was determined first ($R_c$, $R_{c1}$) : We find $\sqrt{\xi}^\star = \sup_j z^*_j \sqrt{\lambda_j}c_j^{-1}$, which serves as the test statistic, and get $p$-value as $1 - \prod_j \Phi_{sym}(z_j^\star)$ where $z_j^{\star} := \frac{c_j}{\sqrt{\lambda_j} } \sqrt{\xi}^{\star}$.
	\item $R_{z}$ : We utilize $1- \alpha = \exp\left( \frac{\sum_k \lambda_k}{\lambda_j} \log \Phi_{sym}(z_j)\right)$ for each $j$ and therefore use \\ $\inf_j \left[ 1 - \exp\left( \frac{\sum_k \lambda_k}{\lambda_j} \log \Phi_{sym}(z^*_j)\right) \right]$ as our p-value. 
	\item $R_{z1}$ : We find $M^\star = \sup_j \sqrt{2\pi}\lambda_jf(z_j^*)$, which serves as the test statistic, and get $p$-value as $1 - \prod_j \Phi_{sym}(z_j^\star)$ where $z_j^\star = f^{-1}(\frac{M^\star}{\sqrt{2\pi}\lambda_j})$.	
\end{enumerate}

For the small sample versions of hyper-rectangular regions, we find $t$-score for each $j$ as $t_j^* := \langle \sqrt{N}(\hat{\theta} - \theta_0), v_j \rangle / \sqrt{\tilde{\lambda}_j}$ and convert them into $z$-scores such that $\Phi_{sym}(z_j^*)=\mbP(|T_j| \leq t_j^*)$ for each $j$. We can then follow the same step for each region as described above.

\section{Smoothing by penalty on $2^{nd}$ derivative and it's covariance} 

Let $\{X_i\}_{i=1}^{N} \subset \mcH = L^2[0,1]$ be an independent sample with mean $\mu$.  Using a basis expansion we can write $X_i \approx \sum_{j=1}^J x_{ij}e_j$ and $\mu = \sum_{j=1}^J m_{j}e_j$ where $\{e_j\}_{j=1}^J$ is a basis of $\mcH$, and satisfy $\|e_j\| = 1$ and the second derivative $e_j^{(2)}$ exists for all $j$.  We assume that $J$ is large so that the degree of smoothing on $\mu$ is controlled primarily through the penalty.  A standard smoothed estimate for $\mu$ can be found by penalizing the second derivative
\begin{align*}
\hat \mu = \arg \min_\mu  N^{-1} \sum_{i=1}^N \| X_i - \mu \|^2 + \lambda \| \mu^{(2)} \| \end{align*}
where $\lambda$ is the smoothing parameter. Note 
\begin{align*}
\| X_i - \mu \|^2 &= \left\| \sum_{j} (x_{ij} - m_j)e_j \right \|^2 =
\left\langle \sum_{j} (x_{ij} - m_j)e_j , \sum_{j'} (x_{ij'} - m_{j'})e_{j'} \right\rangle  \\
&= \sum_j \sum_{j'} ( x_{ij}x_{ij'}  - x_{ij}m_{j'} - x_{ij'}m_j + m_jm_{j'}) \langle e_j, e_{j'} \rangle 
= \bx_i'B\bx_i  - 2 \bx_i' B \bm + \bm' B \bm
\end{align*}  
where $\bx_i = (x_{i1}, \dots , x_{iJ})'$, \ $\bm = (m_{i}, \dots , m_{J})'$, and $B_{jj'} = \langle e_j, e_{j'} \rangle $.
Likewise, 
\begin{align*}
\| \mu^{(2)} \|^2 &= \left\| \sum_{j} m_j e_j^{(2)} \right \|^2 = 
\left\langle \sum_{j} m_j e_j^{(2)}, \sum_{j'} m_{j'} e_{j'}^{(2)} \right\rangle 
= \sum_j \sum_{j'} m_jm_{j'} \langle e_j^{(2)} , e_{j'}^{(2)} \rangle 
= \bm' D \bm
\end{align*}  
where $D_{jj'} = \langle e_j^{(2)}, e_{j'}^{(2)} \rangle $.
Therefore, the target function can be expressed as
\[
N^{-1} \sum_{i=1}^N \left( \bx_i'B\bx_i  - 2 \bx_i' B \bm \right) + \bm' B \bm + \lambda  \bm' D \bm
\]
using familiar matrix notations. By taking $\partial Q / \partial \bm = 0$ and defining $\bar{\bx} := N^{-1}\sum_{i=1}^N \bx_i$, we get
\[
\hat{\bm} = (B + \lambda D)^{-1} B \bar{\bx}
\]
which then gives estimate of $\mu$ as $\hat{\mu} = \sum_{j=1}^J \hat{m}_j e_j$. By defining $A := (B + \lambda D)^{-1} B$, we get the covariance of $\hat{\bm}$ as $\Sigma_{\hat{\bm}} := A \Sigma_{\bar{\bx}} A'$ where $\Sigma_{\bar{\bx}} := \Cov(\bar{\bx})$. Finally, the covariance operator of $\hat{\mu}$, $C_{\hat{\mu}}$, 
takes a bivariate function form of 
$
C_{\hat{\mu}}(t,s) = \sum_j \sum_{j'} \Cov(\hat{m}_j,\hat{m}_{j'})e_j(t)e_{j'}(s)
$, or 
$
C_{\hat{\mu}} = \sum_j \sum_{j'} \Cov(\hat{m}_j,\hat{m}_{j'})(e_j \otimes e_{j'})
$
as an operator. 

\section{Power Tables}

The following four tables come from Subsection 4.1.2 

\label{supplement:PowerTable}
\begin{table}[ht]
	\small
	\caption{$N=25, \ \nu=1/2$}
	\centering
	\begin{tabular}{l|rrrrrrrrrr|r}
		\hline
		\textbf{1. Shift} & 0.02 & 0.04 & 0.06 & 0.08 & 0.1 & 0.12 & 0.14 & 0.16 & 0.18 & 0.2 & Avg. \\ 
		\hline
		$\hat{E}_{norm}$ & 0.09 & 0.16 & 0.29 & 0.47 & 0.63 & 0.78 & 0.88 & 0.96 & 0.98 & 1.00 & 0.62 \\ 
		$\hat{E}_{PC(3)}$ & 0.09 & 0.14 & 0.24 & 0.37 & 0.52 & 0.67 & 0.81 & 0.91 & 0.95 & 0.99 & 0.57 \\ 
		$\hat{E}_{c1}$ & 0.09 & 0.16 & 0.30 & 0.47 & 0.64 & 0.79 & 0.89 & 0.96 & 0.98 & 1.00 & 0.63 \\ 
		$\hat{E}_{c}$ & 0.10 & 0.17 & 0.31 & 0.48 & 0.64 & 0.79 & 0.89 & 0.96 & 0.98 & 1.00 & 0.63 \\ 
		$\hat{R}_{z1}$ & 0.09 & 0.15 & 0.28 & 0.44 & 0.60 & 0.75 & 0.86 & 0.94 & 0.98 & 0.99 & 0.61 \\ 
		$\hat{R}_{z}$ & 0.09 & 0.15 & 0.28 & 0.44 & 0.60 & 0.75 & 0.87 & 0.95 & 0.98 & 0.99 & 0.61 \\ 
		$\hat{R}_{z1s}$ & 0.06 & 0.11 & 0.23 & 0.38 & 0.54 & 0.70 & 0.83 & 0.93 & 0.97 & 0.99 & 0.57 \\ 
		$\hat{R}_{zs}$ & 0.06 & 0.11 & 0.23 & 0.38 & 0.55 & 0.71 & 0.83 & 0.93 & 0.97 & 0.99 & 0.58 \\ 
		$\hat{B}_{s}$ & 0.12 & 0.20 & 0.34 & 0.51 & 0.66 & 0.80 & 0.90 & 0.96 & 0.98 & 1.00 & 0.65 \\ 
		\hline
		\textbf{2. Scale} & 0.02 & 0.04 & 0.06 & 0.08 & 0.1 & 0.12 & 0.14 & 0.16 & 0.18 & 0.2 & Avg. \\ 
		\hline
		$\hat{E}_{norm}$ & 0.06 & 0.09 & 0.13 & 0.20 & 0.30 & 0.41 & 0.54 & 0.68 & 0.79 & 0.88 & 0.41 \\ 
		$\hat{E}_{PC(3)}$ & 0.08 & 0.13 & 0.22 & 0.35 & 0.51 & 0.65 & 0.80 & 0.89 & 0.95 & 0.98 & 0.56 \\ 
		$\hat{E}_{c1}$ & 0.07 & 0.11 & 0.18 & 0.28 & 0.42 & 0.56 & 0.71 & 0.84 & 0.92 & 0.96 & 0.51 \\ 
		$\hat{E}_{c}$ & 0.08 & 0.13 & 0.20 & 0.32 & 0.47 & 0.61 & 0.76 & 0.87 & 0.94 & 0.97 & 0.53 \\ 
		$\hat{R}_{z1}$ & 0.08 & 0.12 & 0.19 & 0.29 & 0.41 & 0.56 & 0.71 & 0.83 & 0.91 & 0.95 & 0.50 \\ 
		$\hat{R}_{z}$ & 0.08 & 0.11 & 0.19 & 0.28 & 0.41 & 0.56 & 0.70 & 0.82 & 0.91 & 0.95 & 0.50 \\ 
		$\hat{R}_{z1s}$ & 0.05 & 0.08 & 0.12 & 0.20 & 0.32 & 0.46 & 0.61 & 0.76 & 0.86 & 0.92 & 0.44 \\ 
		$\hat{R}_{zs}$ & 0.05 & 0.07 & 0.12 & 0.20 & 0.32 & 0.46 & 0.61 & 0.75 & 0.85 & 0.92 & 0.44 \\ 
		$\hat{B}_{s}$ & 0.10 & 0.14 & 0.22 & 0.33 & 0.47 & 0.60 & 0.74 & 0.85 & 0.92 & 0.96 & 0.53 \\ 
		\hline
		\textbf{3. Local Shift} & 0.02 & 0.04 & 0.06 & 0.08 & 0.1 & 0.12 & 0.14 & 0.16 & 0.18 & 0.2 & Avg. \\ 
		\hline
		$\hat{E}_{norm}$ & 0.06 & 0.07 & 0.08 & 0.10 & 0.14 & 0.19 & 0.26 & 0.36 & 0.49 & 0.62 & 0.23 \\ 
		$\hat{E}_{PC(3)}$ & 0.08 & 0.12 & 0.20 & 0.33 & 0.48 & 0.65 & 0.78 & 0.87 & 0.93 & 0.96 & 0.54 \\ 
		$\hat{E}_{c1}$ & 0.07 & 0.10 & 0.18 & 0.35 & 0.63 & 0.89 & 0.98 & 1.00 & 1.00 & 1.00 & 0.62 \\ 
		$\hat{E}_{c}$ & 0.09 & 0.14 & 0.31 & 0.61 & 0.88 & 0.99 & 1.00 & 1.00 & 1.00 & 1.00 & 0.70 \\ 
		$\hat{R}_{z1}$ & 0.09 & 0.17 & 0.40 & 0.69 & 0.91 & 0.99 & 1.00 & 1.00 & 1.00 & 1.00 & 0.72 \\ 
		$\hat{R}_{z}$ & 0.09 & 0.17 & 0.39 & 0.68 & 0.91 & 0.98 & 1.00 & 1.00 & 1.00 & 1.00 & 0.72 \\ 
		$\hat{R}_{z1s}$ & 0.04 & 0.07 & 0.19 & 0.42 & 0.72 & 0.92 & 0.99 & 1.00 & 1.00 & 1.00 & 0.64 \\ 
		$\hat{R}_{zs}$ & 0.04 & 0.07 & 0.18 & 0.41 & 0.71 & 0.91 & 0.98 & 1.00 & 1.00 & 1.00 & 0.63 \\ 
		$\hat{B}_{s}$ & 0.10 & 0.13 & 0.21 & 0.31 & 0.45 & 0.61 & 0.75 & 0.86 & 0.94 & 0.97 & 0.53 \\ 
		\hline
	\end{tabular}
\end{table}

\begin{table}[ht]
	\caption{$N=25, \ \nu=3/2$}
	\centering
	\begin{tabular}{l|rrrrrrrrrr|r}
		\hline
		\textbf{1. Shift} & 0.02 & 0.04 & 0.06 & 0.08 & 0.1 & 0.12 & 0.14 & 0.16 & 0.18 & 0.2 & Avg. \\ 
		\hline
		$\hat{E}_{norm}$ & 0.08 & 0.15 & 0.27 & 0.41 & 0.57 & 0.73 & 0.84 & 0.92 & 0.97 & 0.99 & 0.59 \\ 
		$\hat{E}_{PC(3)}$ & 0.10 & 0.16 & 0.25 & 0.36 & 0.51 & 0.67 & 0.80 & 0.89 & 0.95 & 0.98 & 0.57 \\ 
		$\hat{E}_{c1}$ & 0.09 & 0.16 & 0.27 & 0.42 & 0.58 & 0.74 & 0.85 & 0.93 & 0.97 & 0.99 & 0.60 \\ 
		$\hat{E}_{c}$ & 0.09 & 0.16 & 0.28 & 0.42 & 0.58 & 0.74 & 0.85 & 0.93 & 0.97 & 0.99 & 0.60 \\ 
		$\hat{R}_{z1}$ & 0.09 & 0.15 & 0.26 & 0.41 & 0.56 & 0.72 & 0.84 & 0.92 & 0.96 & 0.99 & 0.59 \\ 
		$\hat{R}_{z}$ & 0.09 & 0.15 & 0.26 & 0.41 & 0.56 & 0.72 & 0.84 & 0.92 & 0.97 & 0.99 & 0.59 \\ 
		$\hat{R}_{z1s}$ & 0.06 & 0.12 & 0.22 & 0.35 & 0.50 & 0.67 & 0.81 & 0.89 & 0.95 & 0.98 & 0.56 \\ 
		$\hat{R}_{zs}$ & 0.06 & 0.12 & 0.22 & 0.36 & 0.51 & 0.67 & 0.81 & 0.90 & 0.95 & 0.98 & 0.56 \\ 
		$\hat{B}_{s}$ & 0.09 & 0.17 & 0.28 & 0.43 & 0.59 & 0.75 & 0.86 & 0.93 & 0.97 & 0.99 & 0.61 \\ 
		\hline
		\textbf{2. Scale} & 0.02 & 0.04 & 0.06 & 0.08 & 0.1 & 0.12 & 0.14 & 0.16 & 0.18 & 0.2 & Avg. \\ 
		\hline
		$\hat{E}_{norm}$ & 0.07 & 0.09 & 0.13 & 0.18 & 0.27 & 0.37 & 0.49 & 0.61 & 0.73 & 0.83 & 0.38 \\ 
		$\hat{E}_{PC(3)}$ & 0.11 & 0.18 & 0.29 & 0.43 & 0.60 & 0.75 & 0.87 & 0.94 & 0.98 & 0.99 & 0.61 \\ 
		$\hat{E}_{c1}$ & 0.08 & 0.12 & 0.19 & 0.30 & 0.45 & 0.62 & 0.77 & 0.89 & 0.95 & 0.98 & 0.53 \\ 
		$\hat{E}_{c}$ & 0.08 & 0.12 & 0.20 & 0.31 & 0.46 & 0.62 & 0.77 & 0.89 & 0.95 & 0.98 & 0.54 \\ 
		$\hat{R}_{z1}$ & 0.08 & 0.13 & 0.21 & 0.33 & 0.48 & 0.64 & 0.79 & 0.89 & 0.95 & 0.98 & 0.55 \\ 
		$\hat{R}_{z}$ & 0.08 & 0.13 & 0.21 & 0.32 & 0.48 & 0.64 & 0.78 & 0.89 & 0.95 & 0.98 & 0.55 \\ 
		$\hat{R}_{z1s}$ & 0.06 & 0.09 & 0.15 & 0.24 & 0.39 & 0.54 & 0.70 & 0.83 & 0.92 & 0.97 & 0.49 \\ 
		$\hat{R}_{zs}$ & 0.06 & 0.09 & 0.15 & 0.24 & 0.38 & 0.53 & 0.70 & 0.82 & 0.92 & 0.96 & 0.49 \\ 
		$\hat{B}_{s}$ & 0.08 & 0.12 & 0.20 & 0.30 & 0.43 & 0.59 & 0.72 & 0.83 & 0.91 & 0.96 & 0.51 \\ 
		\hline
		\textbf{3. Local Shift} & 0.02 & 0.04 & 0.06 & 0.08 & 0.1 & 0.12 & 0.14 & 0.16 & 0.18 & 0.2 & Avg. \\ 
		\hline
		$\hat{E}_{norm}$ & 0.07 & 0.06 & 0.08 & 0.09 & 0.11 & 0.15 & 0.19 & 0.24 & 0.33 & 0.43 & 0.17 \\ 
		$\hat{E}_{PC(3)}$ & 0.13 & 0.28 & 0.52 & 0.77 & 0.93 & 0.99 & 1.00 & 1.00 & 1.00 & 1.00 & 0.76 \\ 
		$\hat{E}_{c1}$ & 0.09 & 0.16 & 0.49 & 0.88 & 0.96 & 0.98 & 0.98 & 0.99 & 1.00 & 1.00 & 0.75 \\ 
		$\hat{E}_{c}$ & 0.09 & 0.17 & 0.48 & 0.87 & 0.96 & 0.98 & 0.98 & 0.99 & 1.00 & 1.00 & 0.75 \\ 
		$\hat{R}_{z1}$ & 0.21 & 0.81 & 0.96 & 0.98 & 0.99 & 1.00 & 1.00 & 1.00 & 1.00 & 1.00 & 0.89 \\ 
		$\hat{R}_{z}$ & 0.20 & 0.81 & 0.96 & 0.98 & 0.99 & 1.00 & 1.00 & 1.00 & 1.00 & 1.00 & 0.89 \\ 
		$\hat{R}_{z1s}$ & 0.10 & 0.63 & 0.94 & 0.97 & 0.98 & 0.99 & 1.00 & 1.00 & 1.00 & 1.00 & 0.86 \\ 
		$\hat{R}_{zs}$ & 0.09 & 0.62 & 0.93 & 0.97 & 0.98 & 0.99 & 1.00 & 1.00 & 1.00 & 1.00 & 0.86 \\ 
		$\hat{B}_{s}$ & 0.09 & 0.13 & 0.21 & 0.31 & 0.46 & 0.62 & 0.75 & 0.87 & 0.94 & 0.98 & 0.54 \\ 
		\hline
	\end{tabular}
\end{table}

\begin{table}[ht]
	\caption{$N=100, \ \nu=1/2$}
	\centering
	\begin{tabular}{l|rrrrrrrrrr|r}
		\hline
		\textbf{1. Shift} & 0.01 & 0.02 & 0.03 & 0.04 & 0.05 & 0.06 & 0.07 & 0.08 & 0.09 & 0.1 & Avg. \\ 
		\hline
		$\hat{E}_{norm}$ & 0.08 & 0.15 & 0.28 & 0.46 & 0.63 & 0.79 & 0.90 & 0.96 & 0.98 & 1.00 & 0.62 \\ 
		$\hat{E}_{PC(3)}$ & 0.08 & 0.12 & 0.20 & 0.34 & 0.51 & 0.67 & 0.82 & 0.91 & 0.96 & 0.99 & 0.56 \\ 
		$\hat{E}_{c1}$ & 0.08 & 0.16 & 0.28 & 0.46 & 0.63 & 0.79 & 0.90 & 0.96 & 0.99 & 1.00 & 0.62 \\ 
		$\hat{E}_{c}$ & 0.09 & 0.16 & 0.28 & 0.46 & 0.63 & 0.79 & 0.90 & 0.96 & 0.99 & 1.00 & 0.62 \\ 
		$\hat{R}_{z1}$ & 0.08 & 0.14 & 0.25 & 0.42 & 0.60 & 0.76 & 0.88 & 0.95 & 0.98 & 0.99 & 0.61 \\ 
		$\hat{R}_{z}$ & 0.08 & 0.14 & 0.26 & 0.42 & 0.60 & 0.76 & 0.88 & 0.95 & 0.98 & 0.99 & 0.61 \\ 
		$\hat{R}_{z1s}$ & 0.07 & 0.13 & 0.24 & 0.41 & 0.58 & 0.75 & 0.87 & 0.94 & 0.98 & 0.99 & 0.60 \\ 
		$\hat{R}_{zs}$ & 0.07 & 0.13 & 0.24 & 0.41 & 0.58 & 0.75 & 0.88 & 0.94 & 0.98 & 0.99 & 0.60 \\ 
		$\hat{B}_{s}$ & 0.09 & 0.15 & 0.28 & 0.45 & 0.61 & 0.77 & 0.89 & 0.95 & 0.98 & 0.99 & 0.62 \\ 
		\hline
		\textbf{2. Scale} & 0.01 & 0.02 & 0.03 & 0.04 & 0.05 & 0.06 & 0.07 & 0.08 & 0.09 & 0.1 & Avg. \\ 
		\hline
		$\hat{E}_{norm}$ & 0.06 & 0.08 & 0.13 & 0.19 & 0.29 & 0.41 & 0.55 & 0.69 & 0.81 & 0.89 & 0.41 \\ 
		$\hat{E}_{PC(3)}$ & 0.07 & 0.11 & 0.20 & 0.33 & 0.50 & 0.65 & 0.79 & 0.90 & 0.96 & 0.98 & 0.55 \\ 
		$\hat{E}_{c1}$ & 0.07 & 0.10 & 0.16 & 0.26 & 0.41 & 0.55 & 0.71 & 0.84 & 0.92 & 0.97 & 0.50 \\ 
		$\hat{E}_{c}$ & 0.07 & 0.11 & 0.18 & 0.29 & 0.45 & 0.60 & 0.75 & 0.86 & 0.94 & 0.98 & 0.52 \\ 
		$\hat{R}_{z1}$ & 0.07 & 0.10 & 0.17 & 0.27 & 0.41 & 0.55 & 0.70 & 0.83 & 0.92 & 0.96 & 0.50 \\ 
		$\hat{R}_{z}$ & 0.07 & 0.10 & 0.17 & 0.26 & 0.41 & 0.54 & 0.70 & 0.83 & 0.92 & 0.96 & 0.50 \\ 
		$\hat{R}_{z1s}$ & 0.06 & 0.09 & 0.15 & 0.25 & 0.38 & 0.52 & 0.68 & 0.81 & 0.91 & 0.96 & 0.48 \\ 
		$\hat{R}_{zs}$ & 0.06 & 0.09 & 0.15 & 0.24 & 0.38 & 0.52 & 0.68 & 0.81 & 0.91 & 0.96 & 0.48 \\ 
		$\hat{B}_{s}$ & 0.07 & 0.10 & 0.18 & 0.28 & 0.43 & 0.56 & 0.71 & 0.83 & 0.91 & 0.96 & 0.50 \\ 
		\hline
		\textbf{3. Local Shift} & 0.01 & 0.02 & 0.03 & 0.04 & 0.05 & 0.06 & 0.07 & 0.08 & 0.09 & 0.1 & Avg. \\ 
		\hline
		$\hat{E}_{norm}$ & 0.05 & 0.06 & 0.08 & 0.10 & 0.12 & 0.17 & 0.24 & 0.36 & 0.49 & 0.68 & 0.23 \\ 
		$\hat{E}_{PC(3)}$ & 0.07 & 0.12 & 0.21 & 0.34 & 0.51 & 0.69 & 0.84 & 0.93 & 0.97 & 0.99 & 0.57 \\ 
		$\hat{E}_{c1}$ & 0.06 & 0.09 & 0.19 & 0.39 & 0.76 & 0.96 & 1.00 & 1.00 & 1.00 & 1.00 & 0.64 \\ 
		$\hat{E}_{c}$ & 0.07 & 0.17 & 0.47 & 0.89 & 1.00 & 1.00 & 1.00 & 1.00 & 1.00 & 1.00 & 0.76 \\ 
		$\hat{R}_{z1}$ & 0.07 & 0.20 & 0.56 & 0.90 & 0.99 & 1.00 & 1.00 & 1.00 & 1.00 & 1.00 & 0.77 \\ 
		$\hat{R}_{z}$ & 0.07 & 0.19 & 0.55 & 0.90 & 0.99 & 1.00 & 1.00 & 1.00 & 1.00 & 1.00 & 0.77 \\ 
		$\hat{R}_{z1s}$ & 0.06 & 0.15 & 0.46 & 0.84 & 0.99 & 1.00 & 1.00 & 1.00 & 1.00 & 1.00 & 0.75 \\ 
		$\hat{R}_{zs}$ & 0.06 & 0.14 & 0.45 & 0.83 & 0.99 & 1.00 & 1.00 & 1.00 & 1.00 & 1.00 & 0.75 \\ 
		$\hat{B}_{s}$ & 0.06 & 0.10 & 0.17 & 0.28 & 0.41 & 0.57 & 0.72 & 0.85 & 0.92 & 0.97 & 0.50 \\ 
		\hline
	\end{tabular}
\end{table}

\begin{table}[ht]
	\caption{$N=100, \ \nu=3/2$}
	\centering
	\begin{tabular}{l|rrrrrrrrrr|r}
		\hline
		\textbf{1. Shift} & 0.01 & 0.02 & 0.03 & 0.04 & 0.05 & 0.06 & 0.07 & 0.08 & 0.09 & 0.1 & Avg. \\ 
		\hline
		$\hat{E}_{norm}$ & 0.08 & 0.13 & 0.26 & 0.41 & 0.58 & 0.73 & 0.85 & 0.93 & 0.97 & 0.99 & 0.59 \\ 
		$\hat{E}_{PC(3)}$ & 0.07 & 0.11 & 0.20 & 0.32 & 0.48 & 0.63 & 0.78 & 0.89 & 0.95 & 0.98 & 0.54 \\ 
		$\hat{E}_{c1}$ & 0.08 & 0.14 & 0.26 & 0.40 & 0.58 & 0.73 & 0.85 & 0.93 & 0.97 & 0.99 & 0.59 \\ 
		$\hat{E}_{c}$ & 0.08 & 0.14 & 0.26 & 0.41 & 0.58 & 0.73 & 0.85 & 0.93 & 0.97 & 0.99 & 0.59 \\ 
		$\hat{R}_{z1}$ & 0.07 & 0.13 & 0.25 & 0.39 & 0.56 & 0.71 & 0.84 & 0.92 & 0.97 & 0.99 & 0.58 \\ 
		$\hat{R}_{z}$ & 0.07 & 0.13 & 0.25 & 0.39 & 0.56 & 0.72 & 0.84 & 0.92 & 0.97 & 0.99 & 0.58 \\ 
		$\hat{R}_{z1s}$ & 0.07 & 0.12 & 0.24 & 0.38 & 0.55 & 0.70 & 0.83 & 0.92 & 0.96 & 0.99 & 0.58 \\ 
		$\hat{R}_{zs}$ & 0.07 & 0.13 & 0.24 & 0.38 & 0.55 & 0.70 & 0.83 & 0.92 & 0.96 & 0.99 & 0.58 \\ 
		$\hat{B}_{s}$ & 0.08 & 0.14 & 0.26 & 0.41 & 0.57 & 0.73 & 0.85 & 0.93 & 0.97 & 0.99 & 0.59 \\ 
		\hline
		\textbf{2. Scale}  & 0.01 & 0.02 & 0.03 & 0.04 & 0.05 & 0.06 & 0.07 & 0.08 & 0.09 & 0.1 & Avg. \\ 
		\hline
		$\hat{E}_{norm}$ & 0.06 & 0.08 & 0.12 & 0.18 & 0.25 & 0.35 & 0.47 & 0.61 & 0.73 & 0.84 & 0.37 \\ 
		$\hat{E}_{PC(3)}$ & 0.07 & 0.13 & 0.24 & 0.40 & 0.56 & 0.73 & 0.87 & 0.95 & 0.98 & 0.99 & 0.59 \\ 
		$\hat{E}_{c1}$ & 0.06 & 0.10 & 0.17 & 0.28 & 0.42 & 0.60 & 0.75 & 0.88 & 0.95 & 0.98 & 0.52 \\ 
		$\hat{E}_{c}$ & 0.06 & 0.10 & 0.17 & 0.28 & 0.43 & 0.60 & 0.76 & 0.89 & 0.95 & 0.98 & 0.52 \\ 
		$\hat{R}_{z1}$ & 0.06 & 0.10 & 0.17 & 0.29 & 0.44 & 0.62 & 0.78 & 0.89 & 0.95 & 0.98 & 0.53 \\ 
		$\hat{R}_{z}$ & 0.06 & 0.10 & 0.17 & 0.29 & 0.44 & 0.61 & 0.77 & 0.89 & 0.95 & 0.98 & 0.53 \\ 
		$\hat{R}_{z1s}$ & 0.06 & 0.09 & 0.16 & 0.28 & 0.42 & 0.59 & 0.75 & 0.88 & 0.95 & 0.98 & 0.52 \\ 
		$\hat{R}_{zs}$ & 0.06 & 0.09 & 0.16 & 0.27 & 0.42 & 0.58 & 0.75 & 0.87 & 0.94 & 0.98 & 0.51 \\ 
		$\hat{B}_{s}$ & 0.07 & 0.10 & 0.18 & 0.29 & 0.42 & 0.57 & 0.71 & 0.84 & 0.91 & 0.96 & 0.50 \\ 
		\hline
		\textbf{3. Local Shift} & 0.01 & 0.02 & 0.03 & 0.04 & 0.05 & 0.06 & 0.07 & 0.08 & 0.09 & 0.1 & Avg. \\ 
		\hline
		$\hat{E}_{norm}$ & 0.05 & 0.06 & 0.07 & 0.08 & 0.10 & 0.13 & 0.16 & 0.22 & 0.29 & 0.37 & 0.15 \\ 
		$\hat{E}_{PC(3)}$ & 0.10 & 0.24 & 0.48 & 0.74 & 0.91 & 0.98 & 1.00 & 1.00 & 1.00 & 1.00 & 0.74 \\ 
		$\hat{E}_{c1}$ & 0.07 & 0.13 & 0.44 & 0.93 & 1.00 & 1.00 & 1.00 & 1.00 & 1.00 & 1.00 & 0.76 \\ 
		$\hat{E}_{c}$ & 0.07 & 0.14 & 0.43 & 0.92 & 1.00 & 1.00 & 1.00 & 1.00 & 1.00 & 1.00 & 0.76 \\ 
		$\hat{R}_{z1}$ & 0.16 & 0.90 & 1.00 & 1.00 & 1.00 & 1.00 & 1.00 & 1.00 & 1.00 & 1.00 & 0.91 \\ 
		$\hat{R}_{z}$ & 0.16 & 0.90 & 1.00 & 1.00 & 1.00 & 1.00 & 1.00 & 1.00 & 1.00 & 1.00 & 0.91 \\ 
		$\hat{R}_{z1s}$ & 0.13 & 0.87 & 1.00 & 1.00 & 1.00 & 1.00 & 1.00 & 1.00 & 1.00 & 1.00 & 0.90 \\ 
		$\hat{R}_{zs}$ & 0.13 & 0.86 & 1.00 & 1.00 & 1.00 & 1.00 & 1.00 & 1.00 & 1.00 & 1.00 & 0.90 \\ 
		$\hat{B}_{s}$ & 0.07 & 0.11 & 0.18 & 0.30 & 0.44 & 0.60 & 0.74 & 0.86 & 0.94 & 0.98 & 0.52 \\ 
		\hline
	\end{tabular}
\end{table}

\end{document}